\documentclass[12pt,onecolumn,romanappendices]{IEEEtran}
\usepackage{subfig}
\usepackage{stfloats}
\IEEEoverridecommandlockouts
\newcommand{\nonl}{\renewcommand{\nl}{\let\nl\oldnl}}% Remove line number for one line
\usepackage{etoolbox}
\AtEndEnvironment{example}{\null\hfill$\square$}%
\AtEndEnvironment{remark}{\null\hfill$\square$}%
\usepackage{xstring}

\usepackage{enumerate}
\usepackage{enumitem}
\usepackage{multirow,array}
\usepackage{cite}
\usepackage{graphicx}
\usepackage{psfrag}
\usepackage{array}
\usepackage{algorithm2e}
\usepackage{amssymb}
\usepackage{amsfonts}
\usepackage{float}
\usepackage{textcomp}
\usepackage{xcolor}
\usepackage{tabu}
\usepackage{array}
\usepackage{tabularx}

\usepackage{subfiles}
\usepackage{subcaption}
\PassOptionsToPackage{caption=false,font=footnotesize}{subfig}

\usepackage{amssymb}% http://ctan.org/pkg/amssymb
\usepackage{pifont}% http://ctan.org/pkg/pifont
\usepackage[amsthm,thmmarks]{ntheorem}
\renewcommand\qedsymbol{$\blacksquare$}

\newcolumntype{P}[1]{>{\centering\arraybackslash}p{#1}}
\newcolumntype{M}[1]{>{\centering\arraybackslash}m{#1}}
\usepackage{multirow}
\usepackage{booktabs}

\usepackage{hyperref}
\usepackage{amsmath}
\usepackage{float}
\usepackage{booktabs}
\usepackage{multirow}
\usepackage{siunitx}
\usepackage{caption}

\newtheorem{conjecture*}{Conjecture}
\newtheorem{lemma}{Lemma}
\newtheorem{observation}{Observation}

\newtheorem{remark}{Remark}
\newtheorem{definition}{Definition}
\newtheorem{theorem}{Theorem}
\newtheorem{example}{Example}

\newcommand{\pkarxiv}[1]{}
\newcommand{\unif}{{\mathcal U}}
\newcommand{\subfalph}{{\mathcal A}}

\newcommand{\wsj}{{\mathcal P}_{\mathsf{WSJ}}}
\newcommand{\wbu}{{\mathcal P}_{\mathsf{WBU}}}

\newcommand{\wcsj}{{\mathcal P}_{T-\mathsf{WSJ}}}
\newcommand{\expec}{{\mathbb E}}
\newcommand{\mileak}{\rho^{\mathsf{MIL}}}
\newcommand{\maxleak}{\rho^{\mathsf{MaxL}}}

\newcommand{\Prob}{\mathsf{P}}

\newcommand{\milq}{q_0}
\newcommand{\milg}{g_{\mathsf{MIL}}}
\newcommand{\milgprime}{\tilde{g}_{\mathsf{MIL}}}
\newcommand{\milh}{h_{\mathsf{MIL}}}
\newcommand{\milf}{f_{\mathsf{MIL}}}
\newcommand{\milC}{C_{\mathsf{MIL}}}
\newcommand{\milk}{k_{\mathsf{MIL}}}
\newcommand{\maxlq}{q_0}
\newcommand{\maxlg}{g_{\mathsf{MaxL}}}
\newcommand{\maxlgprime}{\tilde{g}_{\mathsf{MaxL}}}
\newcommand{\maxlh}{h_{\mathsf{MaxL}}}
\newcommand{\maxlf}{f_{\mathsf{MaxL}}}
\newcommand{\maxlC}{C_{\mathsf{MaxL}}}
\newcommand{\maxlk}{k_{\mathsf{MaxL}}}

\newcommand{\comment}[1]{}
\title{Converse Bounds for Sun-Jafar-type\\ Weak Private Information Retrieval} 

\begin{document}
%%%%%%%%%%%%%%%%%%%%%%%%%%%%%%%%%%%%%%%%%%%%%%%%%%%%%%%%%%%%
% \author{
%   \IEEEauthorblockN{Authors}

% \vspace{-0.2cm}
% %%%%%%%%%%%%%%%%%%%
% }

\author{Chandan Anand, Jayesh Seshadri, Prasad Krishnan, Gowtham R. Kurri % <-this % stops a space
%\thanks{\hrule}%

\thanks{Chandan, Jayesh, Dr. Krishnan and Dr. Kurri are with the Signal Processing and Communications Research Center, International Institute of Information Technology, Hyderabad, 500032, India (email: $\{$chandan.anand@research., jayesh.seshadri@research., prasad.krishnan@,gowtham.kurri@$\}$iiit.ac.in). 
% Acknowledgment:  Dr. Krishnan acknowledges support from ANRF-SERB project CRG/2023/008696. Dr. Kurri acknowledges support from ANRF project ANRF/ECRG/2024/005472/ENS. 
}
}

\maketitle

\allowdisplaybreaks % Command for 

%%%%%%%%%%%%%%%%%%%%%%%%%%%%%%%%%%%%%%%%%%%%%%%%%%%%%%%%%%%%
\begin{abstract}
% THIS PAPER IS ELIGIBLE FOR THE STUDENT PAPER AWARD. 
 
%%% 

% We study the problem of weakly private information retrieval (WPIR), where a client wishes to retrieve a desired file from $N$ servers each storing $M$ files, such that the perfect privacy constraint is relaxed in a controlled manner, allowing for an increase in rate while compromising partial information leakage. Thus, this gives rise to a rate-privacy trade-off in PIR. 
% \pk{Abstract to be revised similar to ISIT version?}\pk{June 1: please tell if it has been revised or not} \chandan{June 1: Need revision.} 
Building on the well-established capacity-achieving schemes of Sun-Jafar (for replicated storage) and the closely related scheme of Banawan-Ulukus (for MDS-coded setting), a recent work by Anand et al. proposed new classes of weak private information retrieval (WPIR) schemes for the collusion-free (replication and MDS-coded) setting, as well as for the $T$-colluding scenario. In their work, Anand et al. characterized the expressions for the rate-privacy trade-offs for these classes of WPIR schemes, under the mutual information leakage and maximal leakage metrics. Explicit achievable trade-offs for the same were also presented, which were shown to be competitive or better than prior WPIR schemes. However, the class-wise optimality of the reported trade-offs was unknown. In this work, we show that the explicit rate-privacy trade-offs reported for the Sun-Jafar-type schemes by Anand et al. are class-wise optimal for the non-colluding and replicated setting. Furthermore, we prove the class-wise optimality for Banawan-Ulukus-type MDS-WPIR and Sun-Jafar-type $T$-colluding WPIR schemes, under threshold-constraints on the system parameters. When these threshold-constraints do not hold, we present counter-examples which show that even higher rates than those reported before can be achieved. 

%%%
\end{abstract}
%TO BE INCLUDED LATER 
% \textit{\small Due to space restrictions, this submission is a shorter version. The extended version of this work is available on ArXiv \cite{Anandetal_ISIT2025_WPIR}, and contains the missing proofs, additional results, examples, and figures.}

% \let\thefootnote\relax\footnotetext{
% Chandan, Jayesh, Dr. Krishnan and Dr. Kurri are with the Signal Processing and Communications Research Center, International Institute of Information Technology, Hyderabad, 500032, India (email: $\{$chandan.anand@research., jayesh.seshadri@research., prasad.krishnan@,gowtham.kurri@$\}$iiit.ac.in). Acknowledgment:  Dr. Krishnan acknowledges support from ANRF-SERB project CRG/2023/008696. Dr. Kurri acknowledges support from ANRF project ANRF/ECRG/2024/005472/ENS. 
% }
% \newpage

\section{Introduction}
\label{sec:intro}
%%%%%%%%%%%%%%%%%%%%%%%%%%%%%%%%%%%%%%%%%%%%%%%%%%%%%%%%%%%%%
% In today’s increasingly interconnected digital age, protecting privacy has become more critical than ever. 
In \textit{Private Information Retrieval} (PIR), a client intends to retrieve a file out of $M$ files stored in $N$ servers (or databases) that store a collection of files, while ensuring that no individual server (or subsets of servers) can determine the index of the file being requested. A query-response protocol issued by the client to the servers enables such private retrieval. This is called a \textit{PIR protocol} (or scheme). In these protocols, a client sends a set of queries to the servers in a manner that masks the identity of the requested file. The server returns the corresponding responses, from which the client reconstructs the desired information.

The performance of a PIR scheme is measured using the retrieval rate, defined as the ratio of the size of the desired file to the total number of bits downloaded during the execution of the scheme. The maximum rate possible across all PIR schemes is known as the capacity of PIR. The seminal work of Sun and Jafar \cite{sun2017capacity} establishes the capacity of PIR schemes for systems with replicated storage, where each server stores all files, and with non-colluding servers, where no two servers are allowed to share information. This work was extended to the $T$-collusion PIR (or simply, $T$-PIR) scenario \cite{sun2017_T_capacity}, in which upto $T\leq N-1$ servers can collude. Banawan and Ulukus \cite{banawan2018capacity} presented the capacity-achieving scheme for collusion-free PIR for the $(K, N)$-MDS storage scenario, in which the files are placed across the $N$ servers upon encoding using a MDS code of dimension $K$ and length $N$. Tian, Sun and Chen \cite{tian2019capacity} presented a capacity-achieving scheme (which we call the \textit{TSC scheme}) for the collusion-free scenario that simultaneously also achieved the minimum message size. 

In a PIR scheme, when the privacy constraint is relaxed in a controlled manner so that no individual server can infer the client’s requested file index beyond a specified threshold, the resulting scheme is defined as a Weak PIR (WPIR) scheme. WPIR schemes have been studied for the collusion-free replicated servers in \cite{lin2021multi,IWPIR2022,WSJPIR2024,CIPM2024} using various privacy metrics. These metrics include but are not limited to \textit{mutual information leakage} ($\mathsf{MIL}$) \cite{lin2021multi}, \textit{maximal leakage} ($\mathsf{MaxL}$) \cite{issa2019operational}, \textit{worst-case information leakage} \cite{kopf2007information}, \textit{$\epsilon$-privacy} \cite{samy2021asymmetric}, and the \textit{converse-induced privacy metric} \cite{jia2019capacity,CIPM2024}. These schemes in these works can be viewed as weak versions of the TSC scheme \cite{tian2019capacity}.  Among these, the work \cite{IWPIR2022} established the best known rate-privacy trade-offs under mutual information and maximal leakage metrics. The work\cite{orvedal2024weakly} extended the TSC-type weak-PIR schemes to MDS-coded servers. A recent work \cite{zhao2025optimizingDP} optimized another version of Weak-PIR, called Leaky-PIR (L-PIR) (studied previously in \cite{samy2021asymmetric,lin2021multi}), under the $\epsilon$-differential privacy metric.   Information-theoretic converses for WPIR under $\mathsf{MIL}$ and $\mathsf{MaxL}$ constraints were presented in \cite{lin2021multi}. These converses apply to a wide-variety of PIR schemes. However, the tightness of these bounds holds only for a very limited choice of system parameters. %The first detailed study of WPIR was shown in \cite{lin2021multi}, where the authors provide a converse bound for an upper bound on the rate of WPIR. 
%Interestingly, the capacity of weak-PIR is first addressed in \cite{CIPM2024}. 

A recent work \cite{WSJPIR2024} by the present authors showed a natural technique by which any capacity-achieving \textit{Sun-Jafar-type} scheme can be extended to an appropriate WPIR scheme under the $\mathsf{MIL}$ and $\mathsf{MaxL}$ privacy metrics. Using this technique, WPIR versions of the original scheme for collusion-free setup \cite{sun2017capacity}, the $T$-collusion scheme \cite{sun2017_T_capacity}, and the MDS-PIR scheme \cite{banawan2018capacity}, were presented \cite{WSJPIR2024}, and rate-privacy trade-offs for these were characterized. The $T$-collusion WPIR scheme in \cite{WSJPIR2024} was the first $T$-collusion WPIR scheme, and remains the only such scheme till the time of this submission, to the best of our knowledge. 
% The work \cite{WSJPIR2024} characterized the rate–privacy trade-offs by comparing its approach with that of \cite{IWPIR2022} under both the $\mathsf{MIL}$ metric and the $\mathsf{MaxL}$ metric. 
% For the collusion-free scenario under replicated servers, the scheme from \cite{WSJPIR2024} has an identical rate-privacy trade-off to the scheme from \cite{IWPIR2022} under the $\mathsf{MaxL}$ metric, while it is suboptimal compared to the scheme from \cite{IWPIR2022} for the $\mathsf{MIL}$ metric.
For the collusion-free scenario, under the $\mathsf{MIL}$ metric, the scheme from \cite{WSJPIR2024} achieves a suboptimal trade-off compared to \cite{IWPIR2022}. However, under the $\mathsf{MaxL}$ metric, the rate-privacy trade-offs are identical to \cite{IWPIR2022} $\mathsf{MaxL}$ result. For the colusion-free and the MDS-coded scenarios, the Sun-Jafar-type schemes \cite{WSJPIR2024} have competitive or better performance than prior work.%For the MDS-coded WPIR setting, the scheme proposed in \cite{WSJPIR2024} performs better than \cite{orvedal2024weakly}, achieving higher rates for the same privacy-leakage constraint. 
% A common drawback of these schemes, one that is shared with Sun-Jafar's original PIR scheme \cite{WSJPIR2024}, is that of large subpacketization or message-size. 

The core idea underlying \cite{WSJPIR2024} is a time-sharing strategy over the number of undesired files involved in the protocol, denoted by $M'$, where $M'$ is a random variable that takes values in set $\{1, 2, \cdots, M-1\}$, according to some distribution $\Prob_{M'}$ chosen by the client. For each PIR scenario, the rate-privacy trade-offs reported in \cite{WSJPIR2024} were achieved by carefully tuning  $\Prob_{M'}$ based on the privacy constraint and the system parameters. However, a complete characterization of the \textit{class-wise optimal} rate-privacy trade-offs that can be achieved by the Sun-Jafar-type WPIR class of schemes given in \cite{WSJPIR2024} is so far unknown. 
% Essentially, this would involve the identification of the class of distributions for the quantity $M'$ that attain the optimal trade-off points. 

In this work, we address this open problem partially. 
 % For the collusion-free, replicated-database setting, we show that the trade-off reported in \cite{WSJPIR2024} is optimal (across all distributions $\Prob_{M'}$) for all system parameters and leakage rates, under both $\mathsf{MIL}$ and $\mathsf{MaxL}$ privacy metrics.
% For the MDS-coded WPIR setting and $T$-WPIR setting, under some constraints on the storage rate $\frac{K}{N}$ and the collusion rate $\frac{T}{N}$ respectively, the rate-privacy trade-offs reported in \cite{WSJPIR2024} are optimal.
The subsequent organization of this work, along with its contributions are as follows.
\begin{itemize}[leftmargin=*]
    \item In Section~\ref{sec:system_model_and_WSJ}, we provide a brief overview of the WPIR system model and recall the construction of the Sun-Jafar-type WPIR schemes \cite{WSJPIR2024}. We then formalize the problem of characterizing the maximum rates of these schemes under given privacy constraints. 
    
    \item In Section~\ref{sec:problem_formulation_main_results}, we formulate the problem of class-wise rate maximization for WPIR protocols over a distributed $(K, N)$-MDS coded storage system within the Banawan-Ulukus framework. This section presents the optimization setup, states the main theorems to be established, and illustrates the results through representative examples. 
    
    \item Theorem~\ref{thm:new_WBU_MIL_thm} characterizes the class-wise optimal rate-privacy trade-off achievable by Banawan-Ulukus-type WPIR schemes under the $\mathsf{MIL}$ constraint. We establish that the trade-offs derived in \cite{WSJPIR2024} are optimal within this class whenever the storage rate satisfies $K/N \le 0.79587$. An analogous result is obtained under the $\mathsf{MaxL}$ constraint through Theorem~\ref{thm:new_WBU_MaxL_thm}, where the corresponding threshold on the storage rate is $K/N \le 0.60199$. Collectively, these results demonstrate that, within the distributed storage setting considered here, the explicit trade-offs in \cite{WSJPIR2024} are class-wise optimal only up to the respective leakage-dependent thresholds.

    \item Through a numerical counter-example, we show that the $(K, N)$-MDS WPIR trade-off reported in \cite{WSJPIR2024} may not be globally optimal across all storage rates under either leakage metric. Specifically, in Example~\ref{example:counterex_higherratebeyondthreshold_MIL}, we consider the $\mathsf{MIL}$ constraint with a storage rate of $K/N=0.9$, where $(N, K)=(10,9)$ and $M=4$ files. %In E w
    We demonstrate that there exists a choice for the distribution ${M'}$ which achieves a strictly higher rate of $0.440$ compared to the rate $0.418$ achieved by the scheme in \cite{WSJPIR2024} at the same leakage level. A similar phenomenon is observed under the $\mathsf{MaxL}$ constraint beyond the threshold $K/N=0.60199$, thereby confirming that the previously known trade-offs cease to be class-wise optimal once the storage rate exceeds the corresponding critical value.

    \item We show that the results obtained for the $(K,N)$-MDS setup extend directly to the $T$-collusion scenario as well. 

    % \item Similar results are also obtained for the $\mathsf{MaxL}$ metric, thus proving class-wise optimality of the $\mathsf{MaxL}$-rate trade-offs in \cite{WSJPIR2024} for the collusion-free replicated-storage scenario, and restricted optimality (with counter-examples) for the $(K, N)$-MDS and $T$-collusion scenarios. 

    \item The proof of Theorem~\ref{thm:new_WBU_MIL_thm} is presented in Section~\ref{sec:proofofThmWBUMIL}, while the proof of Theorem~\ref{thm:new_WBU_MaxL_thm} is provided in Section~\ref{sec:proofofThmWBUMaxL}.
\end{itemize}

We believe that, while this work does not deal with the information-theoretic converses for WPIR under privacy constraints, the results obtained are important, as the trade-offs achieved by the Sun-Jafar-type WPIR schemes in \cite{WSJPIR2024} are the best known (or the only known), especially for the MDS-PIR and the $T$-PIR scenarios. We conclude our work in Section~\ref{sec:conclusion} with a short discussion.

% In Section \ref{sec:main_results_for_MIL}, we derive an upper bound on the achievable rate of the Sun-Jafar type WPIR scheme for a replicated storage, non-colluding server system under fixed mutual information leakage ($\mathsf{MIL}$). These bounds establish the optimality of the distribution of $\Prob_{M'}$ proposed in \cite{WSJPIR2024}. Moreover, we show the optimality of $\Prob_{M'}$ for the collusion-free, MDS-coded setting for a given constraint on the value of $K/N$. We further extend these results to the case of $T$-colluding servers, which also has a constraint on $T/N$. 

%For a fixed $\mathsf{MIL}$ and $\mathsf{MaxL}$, and a constraint on $\frac{K}{N}$, and $\frac{T}{N}$, we provide the upper bound on the rate for the Banawan-Ulukus type WPIR scheme, and for the Sun-Jafar type $T$-collusion WPIR scheme, respectively.
% In Section \ref{sec:optimal_rate_for_MaxL}, we generalize the framework developed in Section \ref{sec:main_results_for_MIL} to incorporate the maximal leakage ($\mathsf{MaxL}$) metric and show the optimality of $\Prob_{M'}$.

\textit{Notations}: For $a, b\in\mathbb{Z}$ such that $a\le b$ we denote the set of integers $\{a, a + 1, \cdots, b\}$ as $[a : b]$. 
We also denote the set $\{1,2\cdots,b\}$ as $[b]$. 
% For a positive integer $n$, we denote $1.2.\cdots.(n-1).n$ as $n!$.
% \chandan{Jan 15: These notations are being inactively used (to state previous work) in this paper.}\pk{Commenting out some part and leaving one notation} 
% For any $x\in\mathbb{R}$, we denote the positive value of $x$ by $(x)_+\triangleq\max(0,x)$. The expectation of the random variable $X$ is denoted by $\mathbb{E}[X]$. 
For a set $A$, the set of all $b$-sized subsets of $A$ is denoted by $\binom{A}{b}$. Let $\mathcal{A}\subset[N]$, where $N$ is positive integer, and $a_1,a_2\cdots,a_N$ are reals, then $a_{\mathcal{A}}\triangleq\{a_i:i\in\mathcal{A}\}$. The notation $\log$ is taken base $2$, while $\ln$ is the natural logarithm. The notation $\exp(a)$ denotes $e^a$.

\section{System model and Review of Weak Sun-Jafar-type Schemes~\cite{WSJPIR2024}}\label{sec:system_model_and_WSJ}

An information retrieval system consists of $N$ servers (or databases), indexed by $DB_0$, $DB_1$,$\cdots$, $DB_{N-1}$. In the replicated-server setup, each server stores a complete copy of the database consisting of $M$ files denoted by $W_1$, $W_2$, $\cdots$, $W_M$, which are independent and identically distributed ($\mathsf{i.i.d}$). Each file has length $L$ and is represented as $W_i=(W_{i, 1}, \cdots, W_{i, L})$, $i \in [1:M]$. We assume that all sub-files $W_{i,j}$, $\forall i,j$, are $\mathsf{i.i.d}$ according to $\unif(\subfalph)$, where $\subfalph$ is an abelian group. Therefore, the joint entropy of the files is given by $H(W_1,\cdots, W_M) = ML \log{ |\mathcal{A}| }$.

A client wishes to retrieve one of the $M$ stored files, denoted by $W_\theta$, where $\theta$ is chosen uniformly at random from $[1:M]$ (we assume $M\geq 2$ throughout). To retrieve the desired file, the client executes a query-response protocol, also known as an information retrieval (IR) scheme. In this protocol, the client generates queries $Q_l^{\theta} : l \in [0:N-1]$, where each query $Q_l^{\theta}$ takes values in the set $\mathcal{Q}_l$. The query $Q_l^{\theta}$ is then sent to server $l$. Upon receiving the query, the server responds with an answer string $A_l^{\theta}$, which is a deterministic function of the query $Q_l^{\theta}$ and the data stored on that server. We denote the length of $A_l^{\theta}$ as $\ell(A_l^{\theta})$. The query-response protocol must satisfy a \textit{correctness} criteria, i.e, enable the client to decode the desired file from the queries and answers. Thus, $H(W_{\theta} | A^{\theta}_{[0:N-1]}, Q^{\theta}_{[0:N-1]}) = 0.$
The rate of an IR protocol is defined as %follows.
%%%
 \begin{align}
 \label{eqn:rate}
     R  \triangleq \frac{L\log|\subfalph|}{\expec[D]}, 
 \end{align}
%%%
where $D=\sum_{l=0}^{N-1} \ell(A_l^{\theta})$
is the total downloaded bits, and the expectation is over the distribution of the queries $Q_{[0:N-1]}^\theta$ and $\theta$.

An IR protocol is said to be $T$-collusion private (or a $T$-PIR scheme), if its queries are designed such that they do not reveal the desired index $\theta$ to any subset of up to $T$ colluding servers. We require the following \textit{privacy constraint} to be satisfied for every ${\cal T}\subseteq [0:N-1]$ such that $|{\cal T}|=T$:
%%%
\begin{align}
\label{eqn:Tcollusionprivacy}
    I(Q^{\theta}_{\mathcal{T}}, A^{\theta}_{\mathcal{T}}, W_{[1:M]}; \theta) = 0.
\end{align}
When $T=1$, the PIR setup is called a \textit{collusion-free} setup.

An information retrieval scheme that is correct but $I(Q^{\theta}_{\mathcal{T}}, A^{\theta}_{\mathcal{T}}, W_{[1:M]}; \theta)>0$ for some $\cal T$, is called a Weak-PIR (WPIR) scheme. We consider the privacy metrics for WPIR based on mutual information leakage ($\mathsf{MIL}$) and maximal leakage ($\mathsf{MaxL}$) measures similar to \cite{zhouetal_2020_WPIR_MaxL,lin2021multi,IWPIR2022}. The MI leakage at server $l$ is $I(Q_l^{\theta};\theta)$, whereas the $\mathsf{MaxL}$ leakage at server $l$ is $$\mathsf{MaxL}(\theta, Q_l^{\theta})\triangleq \log\left(~\sum_{\pmb{q}\in{\mathcal Q}_l}\max_{m\in[1:M]}\Prob_{Q_l^{\theta}|\theta}(\pmb{q}|m)\right).$$ Building upon these, we define the corresponding privacy metrics in WPIR for the collusion-free case as follows.

\begin{definition}\cite{zhouetal_2020_WPIR_MaxL,lin2021multi}
\label{defn:miprivacy}
The MI based privacy metric for a given IR scheme $\cal C$  with queries $Q_l:l\in[0:N-1]$ is defined as 
% follows. 
%%%
% \begin{align}
% \label{eqn:MI_leakage_definition}
   $$\mileak({\cal C}) \triangleq \frac{1}{N}\sum_{l=0}^{N-1} I(\theta;Q_l^{\theta}).$$
 % \end{align}
%%%%%
The $\mathsf{MaxL}$ based privacy metric of IR scheme $\cal C$ is defined as 
% follows.
% \begin{align}
% \label{eqn:defn:maxlprivacy}
$$\maxleak({\cal C})\triangleq\max_{l\in[0:N-1]}\log\left(~\sum_{\pmb{q}\in{\mathcal Q}_l}\max_{m\in[1:M]}\Prob_{Q_l^{\theta}|\theta}(\pmb{q}|m)\right).$$
% \end{align}
\end{definition}
%%%%%
Natural extensions of these to the $T$-collusion setup were defined in \cite{WSJPIR2024} as follows. 
%%%%
%%%

\begin{align*}
% \label{eqn:MI_leakage_definition_Tcoll}
     \mileak({\cal C}) \triangleq \frac{1}{\binom{N}
     {T}}\sum_{{\cal T}\in\binom{[0:N-1]}{T}} I(\theta;Q_{\cal T}^{\theta}),
 \end{align*}
%%%%%

%The maximal leakage-based privacy metric of $T$-collusion WPIR scheme $\cal C$ is defined as follows.
\begin{align*}
\small
% \label{eqn:defn:maxlprivacy_Tcoll}
\maxleak({\cal C})&\triangleq\max_{{\cal T}\in\binom{[0:N-1]}{T}}\log\left(\sum_{\pmb{q}_{\cal T}\in{\mathcal Q}_{\cal T}}\max_{m\in[1:M]}\Prob_{Q_{\cal T}^{\theta}|\theta}(\pmb{q}_{\cal T}|m)\right).
\end{align*}

In the following subsection, we provide a brief description of \cite{WSJPIR2024} for the Weak Sun-Jafar (denoted by $\wsj$) scheme, Weak Banawan-Ulukus ($\wbu$) scheme, and $T$-collusion Weak Sun-Jafar ($\wcsj$) scheme, along with their achieved rate-privacy trade-offs under the mutual information leakage ($\mathsf{MIL}$) metric, and maximal leakage ($\mathsf{MaxL}$) metric. 
%%%
\subsection{Review of Rate-privacy trade-offs from \cite{WSJPIR2024}}
The work \cite{WSJPIR2024} presents three WPIR protocols: (a)  Weak Sun-Jafar (denoted by $\wsj$) scheme collusion-free, replicated storage, (b)  Weak Banawan-Ulukus ($\wbu$) scheme for the $(K,N)$-MDS-coded setup, and (c)  Weak Sun-Jafar ($\wcsj$) scheme for the $T$-collusion scenario. These are obtained by modifying the correspoding prior schemes \cite{sun2017capacity,banawan2018capacity,sun2017_T_capacity} under perfect privacy. Rate-privacy trade-offs achieved by these schemes, under the $\mathsf{MIL}$ and $\mathsf{MaxL}$ metrics  are also presented \cite{WSJPIR2024}. 
%%%

In this subsection, we provide a brief review of the Banawan-Ulukus-WPIR scheme $\wbu$ for MDS storage from \cite{WSJPIR2024}. In this setting, the servers are assumed to be non-colluding, and the $M$ files are stored in the $N$ servers upon encoding using a $K$-dimensional MDS code of length $N$. 

% The trade-offs shown in \cite{WSJPIR2024} for the $\wbu$ scheme essentially carries over to the Weak Sun-Jafar (denoted by $\wsj$) scheme for the case of replicated storage (by fixing $K=1$), as well as the $T$-collusion Weak Sun-Jafar ($\wcsj$) scheme by replacing the term $K$ with $T$. Therefore, we deem it sufficient to present the prior results for this scheme. 

As in the Banawan-Ulukus setting \cite{banawan2018capacity}, the files are considered as matrices of size $N^M \times K$, where each file is encoded using an $(K, N)$-MDS code, which gives us an $N^M \times N$ matrix of encoded symbols. Each column of this matrix is then stored in each of the $N$ servers. Consider that the desired file is $W_k$. The client selects the value of a random variable $M'\in[0:M-1]$ according to some distribution $\Prob_{M'}$. Here, $M'$ denotes the number of undesired files involved in the execution of the protocol. 

If $M'=0$, the client selects $K$ out of $N$ servers uniformly at random and directly demands the desired file from all of the chosen servers. If $M'=m'\in[1:M-1]$, then the client selects a non-empty subset of $J$ of $m'$ undesired file indices uniformly at random from $[1:M]\setminus\{k\}$. The client then executes the Banawan-Ulukus scheme \cite{banawan2018capacity}, for $N$ servers, involving \textit{only} the subset of $\{W_j:j\in J\}\cup\{W_k\}$ of $m'+1$ files. The subpacketization ($N^K$) of the Banawan-Ulukus scheme is suitably large to suit such an execution (see \cite{WSJPIR2024} for more details). 

% Recall that to execute the Banawan-Ulukus scheme with $N$ servers with $M$ files, we need a subpacketization of $KN^M$. However, $N^{m'+1}$ divides $KN^M$ for all $m'\in[1:M-1]$. This ensures that our $\wbu$ can be correctly executed (for more details, please see \cite{WSJPIR2024}).

The $\wbu$ scheme achieves (see \cite[Subsection III.C]{WSJPIR2024}) the following rate, and privacy leakages $\mileak$ and $\maxleak$, for any chosen distribution $\Prob_{M'}$. 
\begin{align}
    R_{WBU}(\Prob_{M'}) &= \frac{1-\frac{K}{N}}{1-\mathbb{E}\left[\left(\frac{K}{N}\right)^{M'+1}\right]}\label{eqn:WBU_rate_expression},\\
% \end{align}
% % The privacy metrics for an arbitrary $\Prob{M'}$ \pk{Jan 4: isn't this arbitrary for the rate also?}\jayesh{resolved} for the $\wbu$ scheme are given by,
% \begin{align}
    \rho^{\mathsf{MIL}}(\Prob_{M'})&=\Prob_{M'}(0)\left(\frac{K\log(M)}{N}\right)\nonumber\\
    &~~~~~+ \sum_{m'=1}^{M-2}\Prob_{M'}(m')\log\left(\frac{M}{m'+1}\right)\label{eqn:WBU_MIL_expression},\\
% \end{align}
% % for $\mathsf{MIL}$.
% \begin{align}
    \rho^{\mathsf{MaxL}}(\Prob_{M'})&=\log\Biggl(\Prob_{M'}(0)\left(\frac{N+KM-K}{N}\right)\nonumber\\ &~~~~~~+\sum_{m'=1}^{M-1}\Prob_{M'}(m')\left(\frac{M}{m'+1}\right)\Biggr)\label{eqn:WBU_MaxL_expression},
\end{align}
where the expectation in the rate is over $\Prob_{M'}$.
Furthermore, the following achievable rate-privacy trade-off was shown in \cite{WSJPIR2024} using the distribution $\Prob_{M'}(0)=\frac{\rho N}{\log(M)K}=1-\Prob_{M'}(M-1)$ in \eqref{eqn:WBU_rate_expression} and \eqref{eqn:WBU_MIL_expression}. %\chandan{June 1: Do we have to change the rate notation of the following expression with subscripts as $WBU,\mathsf{MIL}$ and $WBU,\mathsf{MaxL}$.}\pk{no need for Theorems 1,2}
\begin{theorem}{\cite[Theorem 3]{WSJPIR2024}}
    \label{thm:WBU_MIL}
    For any $\rho\geq 0$, the $\wbu$ scheme achieves the following rate-privacy trade-off.
    \begin{itemize}
        \item $\mileak=\min\left(\rho,\frac{K\log M}{N}\right)$,
        \item $R=\biggl(1+\left(1-\frac{\rho N}{\log(M)K}\right)_+\left(\frac{K}{N} + \cdots + \left(\frac{K}{N}\right)^{M-1}\right)\biggl)^{-1}$,
    \end{itemize}
    where $(x)_{+} \triangleq \max(x,0)$.
\end{theorem}
%%%
\begin{observation}
    \label{obs:for_M_equels_2_MIL}
    For any $N>1$ and $0<K<N$, when $M=2$, this system model restricts $M'$ to support $\{0,1\}$. Thus, $\mileak$ as in \eqref{eqn:WBU_MIL_expression} depends only on $\Prob_{M'}(0)$. Hence, for any prescribed leakage  $\mileak=\rho$, the value of $\Prob_{M'}(0)$ (and therefore $\Prob_{M'}(1)=1-\Prob_{M'}(0)$) is uniquely determined. Consequently, the corresponding PIR rate $R$ as in \eqref{eqn:WBU_rate_expression} is also uniquely fixed as the expression in Theorem~\ref{thm:WBU_MIL}. 
    % Since no alternative $\wbu$ scheme exists that satisfies the same leakage constraint with a different rate, this unique operating point is necessarily optimal within the $\wbu$ class. 
    Therefore, for $M=2$, Theorem~\ref{thm:WBU_MIL} yields the class-wise optimal rate-leakage trade-off. %is non-degenerate and uniquely defines the optimal curve with the WBU class.
\end{observation}

Similarly, an achievable rate-privacy trade-off under the $\mathsf{MaxL}$ metric was obtained using the distribution $\Prob_{M'}(0)=1-\Prob_{M'}(M-1)=\frac{N}{K}\frac{2^{\rho}-1}{M-1}$ in \eqref{eqn:WBU_rate_expression} and \eqref{eqn:WBU_MaxL_expression}. 
\begin{theorem}{\cite[Theorem 4]{WSJPIR2024}}
    \label{thm:WBU_MaxL}
    For any $\rho \geq 0$, the $\wbu$ scheme attains the following rate and $\mathsf{MaxL}$ leakage pairs
    \begin{itemize}
        \item $\maxleak = \min\left(\rho, \log\left(1+K\left(\frac{M-1}{N}\right)\right)\right)$,
        \item $R = \biggl(1+\left(1-\frac{N}{K}\frac{2^{\rho}-1}{M-1}\right)_{+}\left(\frac{K}{N} + \cdots + \left(\frac{K}{N}\right)^{M-1}\right)\biggl)^{-1}$.
    \end{itemize}
\end{theorem} 
\begin{observation}
    \label{obs:for_M_equels_2_MaxL}
    % In the case when $M=2$, the distribution of $M'$ is entirely defined by the single parameter $\Prob_{M'}(0)(=1-\Prob_{M'}(1))$. 
    % Under these conditions,  $\mathsf{MaxL}$ expression in \eqref{eqn:WBU_MaxL_expression} becomes a function of a single parameter $\Prob_{M'}(0)$. 
    If $M=2$, fixing the $\mathsf{MaxL}$ constraint $\mileak=\rho$ in \eqref{eqn:WBU_MaxL_expression}, it is easy to see that the distribution $\Prob_{M'}$ is uniquely determined as a function of $\rho$, as $\Prob_{M'}(0)=1-\Prob_{M'}(1)=N(2^\rho-1)/K$. Thus, the rate $R$ too attains the exact expression in Theorem~\ref{thm:WBU_MaxL} with $M=2$. Thus, the trade-off characterized in Theorem~\ref{thm:WBU_MaxL} is class-wise optimal for $M=2$. %Therefore, the rate-privacy trade-off given in Theorem \ref{thm:WBU_MaxL} already represents the best achievable performance within the $\wbu$ class for $M=2$.
\end{observation}
% Under the maximal leakage metric, the trade-off characterized in Theorem \ref{thm:WBU_MaxL} outperforms the results reported in \cite{orvedal2024weakly} for most values of $N$ and $M$ except for some smaller values. Further, for the MDS-WPIR setting, mutual information as a leakage metric in this context was first introduced in \cite{WSJPIR2024}. However the optimizing distribution of $M'$ for the rate-privacy trade-off is not known from \cite{WSJPIR2024} for the protocol $\wbu$.

 \begin{note*}
     The time-sharing techniques used in the Sun-Jafar-type protocols $\wsj$ and $\wcsj$ for the collusion-free replicated setup and for the $T$-collusion scenario, respectively, are essentially identical to that of the $\wbu$ protocol. Moreover, we note that the expressions for the rate and privacy metrics (as a function of $\Prob_{M'}$) achieved by the Sun-Jafar-type protocol $\wsj$ for the collusion-free replicated setup can be obtained by using $K=1$ in \eqref{eqn:WBU_rate_expression}, \eqref{eqn:WBU_MIL_expression} and \eqref{eqn:WBU_MaxL_expression}. The achievable trade-offs of $\wsj$ shown in \cite{WSJPIR2024} are thus naturally identical to Theorem~\ref{thm:WBU_MIL} (for $\mathsf{MIL}$) and Theorem~\ref{thm:WBU_MaxL} (for $\mathsf{MaxL}$) with $K=1$. Also, the equivalent results for the $T$-collusion WPIR protocol $\wcsj$ can be obtained by replacing $K$ with the collusion parameter $T$. For more details regarding these protocols and the associated results, please see \cite{WSJPIR2024}. 
 \end{note*} 
\section{Problem Formulation and main results}
\label{sec:problem_formulation_main_results}
Our objective is to characterize the maximum achievable rate within the class of Sun-Jafar-type WPIR schemes subject to a specific privacy leakage constraint. Let $\mathcal{P}$ denote the space of all valid probability distributions $\Prob_{M'}$ over the number of undesired files $M' \in [0:M-1]$. For a specific WPIR setting (replicated, MDS, or $T$-collusion) and a privacy metric $\mathsf{Metric} \in \{\mathsf{MIL}, \mathsf{MaxL}\}$, the retrieval rate $R$ and the leakage $\rho^{\mathsf{Metric}}$ are functions of the distribution $\Prob_{M'}$.
We formulate the rate maximization problem as follows:%\chandan{June 1: Do we need to change rate notation as $R_{WBU,\mathsf{Metric}}$}\pk{June 22: No need, as this is a generic statement which is applicable for all , including collusion (you can see). WBU is applicable only for MDS  }
\begin{equation}
\label{eqn:general_optimization}
\begin{aligned}
\max_{\Prob_{M'} \in \mathcal{P}} \quad & R(\Prob_{M'}) \\
\textrm{subject to} \quad & \rho^{\mathsf{Metric}}(\Prob_{M'}) \le \bar{\rho}, \\
& \sum_{m'=0}^{M-1} \Prob_{M'}(m') = 1, \quad \Prob_{M'}(m') \ge 0,
\end{aligned}
\end{equation}
where $\bar{\rho}$ is the allowable leakage threshold. In the subsequent sections, we solve this optimization problem for the specific rate and leakage functions corresponding to each WPIR scheme, assuming an additional threshold condition in the case of the MDS-coded servers scenario. % and the $T$-collusion scenarios. 

\subsection{Main Result: On Classwise Optimality of Sun-Jafar-type schemes \cite{WSJPIR2024} under both $\mathsf{MIL}$ and $\mathsf{MaxL}$ constraints}
\label{sec:main_results_for_MIL}
In this subsection, we will prove 
% \pk{Feb 21: Not ANALYSE but `prove' something. Needs revision } \chandan{Feb 21: Done.}
the optimal rate-privacy trade-off of $\wbu$ scheme under the $\mathsf{MIL}$ constraint. Since Observations~\ref{obs:for_M_equels_2_MIL} and \ref{obs:for_M_equels_2_MaxL} show the class-wise optimality of the trade-offs in Theorem~\ref{thm:WBU_MIL} and \ref{thm:WBU_MaxL} for the $M=2$ case, we assume throughout this section that $M>2$. 
% \pk{Feb 21: Since Observations \ref{obs:for_M_equels_2_MIL} and \ref{obs:for_M_equels_2_MaxL} show the class-wise optimality of the trade-offs in Theorem \ref{thm:WBU_MIL} and \ref{thm:WBU_MaxL} for the $M=2$ case, we assume throughout this section that $M>2$.}\chandan{Feb 21: Understood} %We first focus on the $\wsj$ scheme. Subsequently, we will extend our analysis to the $\wbu$ scheme and $\wcsj$ scheme. For each case, w
We identify the specific distribution of ${M'}$ that maximizes the rate by solving the optimization problem posed in \eqref{eqn:general_optimization}. %distribution, which maximizes the rate for the $\wsj$ scheme, $\wbu$ scheme, and $\wcsj$ scheme for a given constraint on $\rho^{\mathsf{MIL}}$. 

We solve the optimization problem defined in \eqref{eqn:general_optimization} for the $\wbu$ scheme under both $\mathsf{MIL}$ and $\mathsf{MaxL}$ privacy constraints, assuming a threshold condition on the code-rate $K/N$. Specifically, we demonstrate that if $K/N$ is below the specific threshold, the rate-privacy trade-off achieved in Theorem~\ref{thm:WBU_MIL} and Theorem~\ref{thm:WBU_MaxL} are class-wise optimal. The formal result we prove is as follows.

\begin{theorem}
    \label{thm:new_WBU_MIL_thm}
    For $\frac{K}{N}\le0.79587$, $M>2$,
    % \pk{Feb 21: $M>2$}
    and any $\rho\in\left[0,\frac{K\log(M)}{N}\right]$, the maximal rate of the $\wbu$ scheme under an $\mathsf{MIL}$ constraint $\mileak\leq \rho$ is %\pk{June 3: SEE THE CHANGED NOTATION FOR $R^{*}_{WBU}$, and change everywhere in Section \ref{sec:proofofThmWBUMIL}}
    \begin{equation*}
        R^{*}_{WBU,\mathsf{MIL}}=\left(1+\left(1-\frac{N\rho}{K\log(M)}\right)\left(\frac{K}{N}+\cdots+\left(\frac{K}{N}\right)^{M-1}\right)\right)^{-1}
    \end{equation*} if and only if the distribution $\Prob_{M'}$ is given as \begin{itemize}
        \item $\Prob_{M'}(0)=\frac{\rho N}{K\log M}=1-\Prob_{M'}(M-1)$ and $\Prob_{M'}(m')=0$ for all $m'\in[1:M-2]$. 
    \end{itemize}
\end{theorem}
%%%
\begin{theorem}
    \label{thm:new_WBU_MaxL_thm}
    For $\frac{K}{N}\le0.60199$, $M>2$, 
    % \pk{Feb 21: $M>2$}, 
    and any $\rho\in\left[0,\log\left(1+\frac{N(2^{\rho}-1)}{K(M-1)}\right)\right]$, the maximal rate of the $\wbu$ scheme under an $\mathsf{MaxL}$ constraint $\maxleak\leq \rho$ is %\pk{June 3: SEE THE CHANGED NOTATION FOR $R^{*}_{WBU}$, and change everywhere in Section \ref{sec:proofofThmWBUMaxL}}
    \begin{equation*}
        R^{*}_{WBU,\mathsf{MaxL}}=\left(1+\left(1-\frac{N(2^{\rho}-1)}{K(M-1)}\right)\left(\frac{K}{N}+\cdots+\left(\frac{K}{N}\right)^{M-1}\right)\right)^{-1}
    \end{equation*} if and only if the distribution $\Prob_{M'}$ is given as %\pk{Apr 11: Distribution wrong below?}\chandan{Apr 11: corrected.} 
    \begin{itemize}
        \item $\Prob_{M'}(0)=\frac{N(2^{\rho}-1)}{K(M-1)}=1-\Prob_{M'}(M-1)$ and $\Prob_{M'}(m')=0$ for all $m'\in[1:M-2]$. 
    \end{itemize}
\end{theorem}
% \chandan{Feb 22: Theorem statements are fixed.}

The proof of Theorem \ref{thm:new_WBU_MIL_thm} and Theorem \ref{thm:new_WBU_MaxL_thm} are given in Section \ref{sec:proofofThmWBUMIL} and Section \ref{sec:proofofThmWBUMaxL} respectively. The sketch of the proof is as follows. We observe from \eqref{eqn:WBU_rate_expression}, that maximizing the rate is equivalent to maximizing the expectation term that appears in the denominator of the rate expression, which is linear in terms of $\Prob_{M'}(m')$. We therefore treat this linear function as the objective function, subject to the constraints specified in \eqref{eqn:general_optimization}. By characterizing the distribution of $M'$ that maximizes this objective under the given constraint, we obtain the desired result.

% \begin{note*}
% Theorem \ref{thm:new_WBU_MIL_thm}, and Theorem \ref{thm:new_WBU_MaxL_thm}, along with the threshold condition, carries over to the $T$-collusion scheme $\wcsj$ (with only the difference being that the term $T$ replaces the term $K$). The constraint in the $T$-collusion case becomes $T/N\leq 0.534$. On the other hand, for the $\wsj$ scheme, the storage-rate threshold holds automatically, as $K=1$ and $N\geq 2$, which implies $K/N\leq 0.5 < 0.534$, as required by Theorem \ref{thm:new_WBU_MIL_thm}. Thus, the rate-privacy tradeoff achieved by the collusion-free protocol $\wsj$ in \cite{WSJPIR2024} (obtained from Theorem \ref{thm:WBU_MIL} plugging $K=1$) are shown to be optimal.
% \end{note*}

\begin{note*}
Theorems~\ref{thm:new_WBU_MIL_thm} and~\ref{thm:new_WBU_MaxL_thm}, together with their respective threshold conditions, extend to the $T$-collusion scheme $\wcsj$, with the sole modification that $T$ replaces $K$ in the formulation. The constraint for the $T$-collusion case becomes $T/N \leq 0.79587$ for Theorem~\ref{thm:new_WBU_MIL_thm} and $T/N \leq 0.60199$ for Theorem~\ref{thm:new_WBU_MaxL_thm}.%, where $\alpha_{\text{MIL}} = 0.7828$ and $\alpha_{\text{MaxL}} = 0.602$ denote the respective threshold parameters. 

In contrast, for the collusion-free scheme $\wsj$, the storage-rate threshold is satisfied trivially. Since $K=1$ and $N\geq 2$, we have $K/N \leq 0.5$, which is strictly less than both $0.79587$ and $0.60199$, thereby satisfying the requirements of both theorems. Consequently, the rate-privacy tradeoff achieved by the collusion-free protocol $\wsj$ presented in~\cite{WSJPIR2024} (obtained by specializing Theorem~\ref{thm:WBU_MIL} with $K=1$) is shown to be optimal.
\end{note*}

\subsection{Examples illustrating Non-Optimality of Rate-Privacy trade-offs in Theorems~\ref{thm:WBU_MIL} and \ref{thm:WBU_MaxL}}
We now present an example showing that for certain code-rates above the threshold $0.79587$ in Theorem~\ref{thm:WBU_MIL}, we can achieve a rate that is higher than the rate characterized by Theorem~\ref{thm:WBU_MIL}. %\chandan{For $\mathsf{MIL}$ and $\mathsf{MaxL}$ Examples~\ref{example:counterex_higherratebeyondthreshold_MIL}, Example~\ref{example:counterex_higherratebeyondthreshold_MaxL}, and Remark~\ref{rem:WBU_remark_MIL} are updated.}

\begin{example}
\label{example:counterex_higherratebeyondthreshold_MIL}
% \pk{DOUBLE TRIPLE CHECK THIS!}
    Consider a $(K=4,N=5)$-MDS coded WPIR system consisting of $N=5$ servers, and $M=4$ files. In this configuration, the code-rate is $\frac{K}{N}=0.8$, which is strictly greater than the threshold code-rate $0.79587$ mentioned in Theorem \ref{thm:new_WBU_MIL_thm}. Note that  $M'\in[0:3]$ in this case. We present two schemes, both achieving the $\mathsf{MIL}$ metric $\rho=1.333$. 
    \begin{itemize}[leftmargin=*]
    \item As per Theorem \ref{thm:WBU_MIL}, the protocol $\wbu$ obtained by choosing the distribution $\Prob_{M'}(0)=1-\Prob_{M'}(3)=\frac{N\rho}{K\log(M)}=\frac{5\times1.333}{4\log(4)}=0.833$, 
    attains the rate which can be calculated using \eqref{eqn:WBU_rate_expression}, which is $R_{WBU,\mathsf{MIL}}(\Prob_{M'})=0.7542$. Plugging the value of $\Prob_{M'}$ in \eqref{eqn:WBU_MIL_expression} we can get $\mileak=1.333$. 
    \item Now, consider a different distribution for $M'$, given by:  $\mathsf{Q}_{M'}(0)=0.779$, $\mathsf{Q}_{M'}(1)=0$, $\mathsf{Q}_{M'}(2)=0.21$, and $\mathsf{Q}_{M'}(3)=0.11$. From \eqref{eqn:WBU_rate_expression}, the rate of the $\wbu$ scheme with distribution $Q_{M'}$ is $R_{WBU,\mathsf{MIL}}(\mathsf{Q}_{M'})=0.7554$. Putting the values of $\mathsf{Q}_{M'}$ in \eqref{eqn:WBU_MIL_expression} we will get $\mileak=1.333$. 
    % from  where we gets $D_1=-0.0136$ and $D_2=0.0011>0$, thus consider a distribution $\mathsf{Q}_{M'}(0)=a, \mathsf{Q}_{M'}(2)=b$ and $\mathsf{Q}_{M'}(3)=c$. Choose triplet $(a,b,c)=(0.2795,0.6635,0.0570)$ and calculate
    \end{itemize}
    It is clear that $R_{WBU,\mathsf{MIL}}(\mathsf{Q}_{M'})>R_{WBU,\mathsf{MIL}}(\Prob_{M'})$. Thus, in a $\wbu$ scheme for a given system model, if the code-rate is greater than the threshold, then there is room for improving the maximum possible rate, for a fixed $\mathsf{MIL}$ leakage.
\end{example}
\begin{remark}
    \label{rem:WBU_remark_MIL}
    We remark on the choice of the distribution $Q_{M'}$ in Example \ref{example:counterex_higherratebeyondthreshold_MIL}. This distribution $Q_{M'}$ is designed by carefully observing the values $m'$ for which the quantities $D_{m'}:m'\in[1:M-2]$ (in the proof of Theorem \ref{thm:new_WBU_MIL_thm}) become strictly positive. From the arguments in the proof, such scenarios potentially lead to a different distribution for $M'$ yielding higher rates than Theorem \ref{thm:WBU_MIL}. Specifically, for the parameters in Example \ref{example:counterex_higherratebeyondthreshold_MIL} with $\mathsf{MIL}$ metric being $1.333$ and $K/N=0.8>0.79587$, the values $D_{m'}:m'\in\{1,2\}$ can be obtained from \cite[equation (16)]{WSJPIR2024} as $D_1=-0.0136$ and $D_2=0.0011$. Since $D_2>0$, it makes sense to consider a distribution $\mathsf{Q}_{M'}$ with support $\{0,2,3\}$ (note that the distribution in Theorem \ref{thm:WBU_MIL} has support only $\{0,3\}$). The exact choice of  $Q_{M'}$ in Example~\ref{example:counterex_higherratebeyondthreshold_MIL} is chosen by observation. 
\end{remark}

Similar to Example \ref{example:counterex_higherratebeyondthreshold_MIL} for $\mathsf{MIL}$, we observe the same in the $\mathsf{MaxL}$ case as follows. 

\begin{example}
    \label{example:counterex_higherratebeyondthreshold_MaxL}

    % Consider the setting %as in Example \ref{example:counterex_higherratebeyondthreshold_MIL}, 
    Consider the setting with  $N=3, K=2$, and $ M=4$. Thus, the rate of the storage-code is $2/3$, which is greater than the threshold in Theorem \ref{thm:new_WBU_MaxL_thm}. We present two distributions for $M'$, both having the same $\mathsf{MaxL}$ metric $\rho=1.039$, however with different PIR rates. %\pk{WE HAVE TO CALCULATE the PRIVACY METRIC VALUE ALSO AND SHOW IT IS SAME. CAN REMOVE ITEMIZE IF YOU WANT TO. SAME FOR OTHER EXAMPLE}
    
    \begin{itemize}
        \item As per Theorem~\ref{thm:WBU_MaxL}, the $\wbu$ scheme with the distribution $\Prob_{M'}(0)=1-\Prob_{M'}(3) =\frac{N}{K}\frac{2^{\rho}-1}{M-1}=0.533$, achieves the rate, as per \eqref{eqn:WBU_rate_expression}, as $R_{WBU,\mathsf{MaxL}}(\Prob_{M'})=0.6076$. Plugging the values of $\Prob_{M'}$ into \eqref{eqn:WBU_MaxL_expression} and we see that $\maxleak=1.039$. 
        \item We now consider another distribution for $M'$ given by:
        $\mathsf{Q}_{M'}(0)=0.479$, $\mathsf{Q}_{M'}(1)=0.107$, $\mathsf{Q}_{M'}(2)=0$, and $\mathsf{Q}_{M'}(3)=0.414$. From \eqref{eqn:WBU_rate_expression}, we calculate the rate with this distribution as $R_{WBU,\mathsf{MaxL}}(\mathsf{Q}_{M'})=0.6086$. Using the values of $\mathsf{Q}_{M'}$ in \eqref{eqn:WBU_MaxL_expression} gives $\maxleak=1.039$.
    \end{itemize}
    It is evident that $R_{WBU,\mathsf{MaxL}}(\mathsf{Q}_{M'})>R_{WBU,\mathsf{MaxL}}(\Prob_{M'})$, while the $\mathsf{MaxL}$ leakage is the same.  Thus, the trade-off achieved in Theorem \ref{thm:WBU_MaxL} by the $\wbu$ scheme is not optimal for all values of the storage-rate $K/N$. 
\end{example}

% \section{Proof of Theorem \ref{thm:new_WBU_MIL_thm}: \pk{Class-wise Optimality of Theorem \ref{thm:WBU_MIL} under MIL}}
% \label{sec:proofofThmWBUMIL}
\section{Proof of Theorem \ref{thm:new_WBU_MIL_thm}: Optimality of Theorem \ref{thm:WBU_MIL} under $\mathsf{MIL}$}
\label{sec:proofofThmWBUMIL}
In this section, we will provide the proof of Theorem \ref{thm:new_WBU_MIL_thm}. Since the maximal achievable rate $R_{WBU,\mathsf{MIL}}^*$ in the statement of Theorem \ref{thm:new_WBU_MIL_thm} is increasing with $\rho$, it is sufficient to establish the result for  $\mileak=\rho$. From Theorem \ref{thm:WBU_MIL}, the distribution: $\Prob_{M'}(0)=\frac{N\rho}{K\log(M)}=1-\Prob_{M'}(M-1)$ and $\Prob_{M'}(m')=0$ for all $m'\in[1:M-2]$,  achieves the rate $R_{WBU,\mathsf{MIL}}(\Prob_{M'})=R_{WBU,\mathsf{MIL}}^{*}$, while also attaining $\mileak(\Prob_{M'})=\rho$. However, it remains to show that this is the unique distribution that attains the maximum rate under the same constraint $\mileak=\rho$. This will complete the proof of the theorem.

Let $\mathbf{p}=(p(0), p(1), \cdots, p(M-1))$ be an arbitrary probability distribution on $[0:M-1]$ which satisfies the leakage constraint, i.e., $\mileak(\mathbf{p})=\rho$. Thus, from \eqref{eqn:WBU_MIL_expression}, we have 
    \begin{equation}
        \rho=p(0)\frac{K\log(M)}{N}+\sum_{m'=1}^{M-2}p(m')\log\left(\frac{M}{m'+1}\right)\label{eqn:WBU_rho_value_thm}
    \end{equation}

    Define $b_0=\frac{K\log(M)}{N}$ and $b_{m'}=\log\left(\frac{M}{m'+1}\right)$ for $m'\in[1:M-2]$. Using \eqref{eqn:WBU_rho_value_thm} we can write $p(0)$ as following
    \begin{equation}
        p(0)=\frac{1}{b_0}\left(\rho-\sum_{m'=1}^{M-2}p(m')b_{m'}\right)\label{eqn:WBU_p(0)_value}
    \end{equation}
    
    From \eqref{eqn:WBU_rate_expression}, we can conclude that, maximizing the rate of $\wbu$ scheme (or $R_{WBU,\mathsf{MIL}}$) is equivalent to maximizing $T(\mathbf{p})\triangleq\mathbb{E}\left[\left(\frac{K}{N}\right)^{M'+1}\right]=\sum_{m'=0}^{M-1}p(m')c_{m'}$ where $c_{m'}=\left(\frac{K}{N}\right)^{m'+1}$.

    Now, denote $T^*\triangleq c_{M-1} + \frac{\rho}{b_0}(c_0 - c_{M-1})=\left(\frac{K}{N}\right)^{M}+\frac{N\rho}{\log M}\left(\frac{K}{N}-\left(\frac{K}{N}\right)^{M}\right)$. 
Then, from the expression given for $R^*_{WBU,\mathsf{MIL}}$ in Theorem \ref{thm:new_WBU_MIL_thm}, we get $R^*_{WBU,\mathsf{MIL}}=\frac{1-\frac{K}{N}}{1-T^*}$, as the following sequence shows. %\pk{Jan 7: Check this please}\chandan{It's correct.}
\begin{align*}
    R^{*}_{WBU,\mathsf{MIL}}
    &= \left(1+\left(1-\frac{N\rho}{K\log M}\right) \frac{\frac{K}{N}-\left(\frac{K}{N}\right)^{M}}{1-\frac{K}{N}}\right)^{-1} \nonumber \\
    &= \left(\frac{(1-\frac{K}{N}) + (1-\frac{N\rho}{K\log M})\left(\frac{K}{N}-\left(\frac{K}{N}\right)^{M}\right)}{1-\frac{K}{N}}\right)^{-1} \nonumber \\
    &= \left(\frac{1 - \frac{K}{N} + \frac{K}{N} - \left(\frac{K}{N}\right)^{M} - \frac{N\rho}{K\log M}(\frac{K}{N}-\left(\frac{K}{N}\right)^{M})}{1-\frac{K}{N}}\right)^{-1} \nonumber \\
    &= \left(\frac{1 - \left[\left(\frac{K}{N}\right)^{M} + \frac{N\rho}{K\log M}\left(\frac{K}{N}-\left(\frac{K}{N}\right)^{M}\right)\right]}{1-\frac{K}{N}}\right)^{-1} \nonumber \\
    &= \left(\frac{1-T^*}{1-K/N}\right)^{-1} = \frac{1-K/N}{1-T^*}
\end{align*}
Thus, setting $T=T^*$ achieves rate $R^*_{WBU,\mathsf{MIL}}$. Now, substituting $p(0)$ from \eqref{eqn:WBU_p(0)_value}, and using the fact that $p({M-1})=1-\sum_{m'=0}^{M-2}p(m')$, we can write the following. 
    \begin{align} 
        T(\mathbf{p}) &=\sum_{m'=0}^{M-2} p(m') c_{m'} +\left(1-\sum_{m=0}^{M-2}p(m')\right)c_{M-1}\nonumber\\
        &=c_{M-1}+p(0)\left(c_{0}-c_{M-1}\right) + \sum_{m'=1}^{M-2}p(m')\left(c_{m'}-c_{M-1}\right)\nonumber\\
        &=c_{M-1}+\frac{1}{b_0}\left(\rho-\sum_{m'=1}^{M-2}p(m')b_{m'}\right)\left(c_{0}-c_{M-1}\right)\nonumber\\
        &~~~~~~+\sum_{m'=1}^{M-2}p(m')\left(c_{m'}-c_{M-1}\right)\nonumber\\
        &=\left[ c_{M-1} + \frac{\rho}{b_0}(c_0 - c_{M-1}) \right]\nonumber\\
        &~~~~+ \sum_{m'=1}^{M-2} p(m') \underbrace{\left[ (c_{m'} - c_{M-1}) - \frac{b_{m'}}{b_0}(c_0 - c_{M-1}) \right]}_{D_{m'}}\nonumber\\
        &= T^* + \sum_{m'=1}^{M-2} p(m') \underbrace{\left[ (c_{m'} - c_{M-1}) - \frac{b_{m'}}{b_0}(c_0 - c_{M-1}) \right]}_{D_{m'}}. \label{eq:T_expanded}
    \end{align}
    % \chandan{7 Jan: will modify a bit.}

    Now, if we can show that the terms $D_{m'}:m'\in[1:M-2]$ are negative, then, clearly, to maximize $T(\mathbf{p})$, we must fix  $p(m')=0, \forall m'\in[1:M-2]$. This will prove that the support of $\Prob_{M'}$ that maximizes $T(\mathbf{p})$ must be only $\{0,M-1\}$. Amongst all such distributions, the only distribution which gives $\mileak=\rho$ is the distribution $\Prob_{M'}(0)=\frac{\rho N}{K\log M}=1-\Prob_{M'}(M-1)$. Thus, showing that $D_{m'}:m'\in[1:M-2]$ are negative will complete the proof, and this is what we do now.
    
    \begin{align}
        D_{m'}&= (c_{m'} - c_{M-1}) - \frac{b_{m'}}{b_0}(c_0 - c_{M-1})\nonumber\\
        &=\left(\left(\frac{K}{N}\right)^{m'+1}-\left(\frac{K}{N}\right)^{M}\right) - \frac{\log(\frac{M}{m'+1})}{\left(K\frac{\log(M)}{N}\right)}\left(\frac{K}{N}-\left(\frac{K}{N}\right)^{M}\right)\nonumber\\
        &=\left(\left(\frac{K}{N}\right)^{m'+1}-\left(\frac{K}{N}\right)^{M}\right) - \frac{N\log(\frac{M}{m'+1})}{K\log(M)}\left(\frac{K}{N}-\left(\frac{K}{N}\right)^{M}\right)\nonumber\\
        &=\left(\left(\frac{K}{N}\right)^{m'+1}-\left(\frac{K}{N}\right)^{M}\right) - \left(1-\frac{\log(m'+1)}{\log(M)}\right)\left(1-\left(\frac{K}{N}\right)^{M-1}\right)\nonumber\\
        &=\left(1-\left(\frac{K}{N}\right)^{M-1}\right)\Biggl(\frac{\left(\left(\frac{K}{N}\right)^{m'+1}-\left(\frac{K}{N}\right)^{M}\right)}{\left(1-\left(\frac{K}{N}\right)^{M-1}\right)}- \left(1-\frac{\log(m'+1)}{\log(M)}\right)\Biggr).\label{eqn:WBU_MIL_D_m'_criteria}
    \end{align}

    To prove the desired that $D_{m'}$ is strictly negative, it is sufficient to prove that the function %\pk{Apr 12: use $g_{\mathsf{MIL}}$ and similarly for $h$} \chandan{April 12: I used a set of macros for notations like $\milg$, $\milgprime$, $\milh$, etc.}
    \begin{equation}
        \milg(q, m',M)\triangleq\frac{q^{m'+1}-q^{M}}{1-q^{M-1}}- \left(1-\frac{\log(m'+1)}{\log(M)}\right).\label{eqn:MIL_g_function}
    \end{equation}
    is strictly negative for $q\le0.79587$, $M>2$ and $m'\in[1:M-2]$.
    % In what follows, we write the proof in two parts:
    % first for the range $3\le M\le 14$, and then for the remaining case $M>14$.
    
    We begin by establishing the monotonicity of $\milg$ with respect to $q$. For fixed $m'$ and $M$, define \begin{equation}
    \label{eqn:hMIL_auxfunction_defn}
        \milh(q)\triangleq q^{m'+1}\left(\frac{1-q^{M-m'-1}}{1-q^{M-1}}\right).\nonumber
    \end{equation}
    Then $\milg(q,m',M)=\milh(q)-(1-\frac{\ln(m'+1)}{\ln(M)})$, and except $\milh(q)$ remaining terms in $\milg$ are independent of $q$, so monotonicity of $\milg$ in $q$ is equivalent to monotonicity of $\milh$ in $q$.
    Now, differentiating $\milh(q)$ with respect to $q$ we have \begin{equation}
        \frac{\partial \milh(q)}{\partial q}=\underbrace{(m'+1)q^{m'}\frac{1-q^{M-m'-1}}{1-q^{M-1}}}_{T_1}+q^{m'+1}\underbrace{\frac{\partial}{\partial q}\left(\frac{1-q^{M-m'-1}}{1-q^{M-1}}\right)}_{T_2}\label{eqn:g_increases_with_q}.
    \end{equation}

    % \chandan{For a fixed $M>2$ and $m'\in[1:M-2]$, and $q\in(0,1)$, the term $T_1$ of \eqref{eqn:g_increases_with_q} is positive. Term $T_2$ is positive which can be proved using Lemma~\ref{lemma:f_increases_with_x_WBU_MIL_inter}. Positivity of term $T_1$ and $T_2$ implies $\frac{\partial \milh(q)}{\partial q}$ is positive. } 
    For a fixed $M>2$ and $m'\in[1:M-2]$, and  $q\in(0,1)$, the term $T_1$ of \eqref{eqn:g_increases_with_q} is positive. Further, using Lemma~\ref{lemma:f_increases_with_x_WBU_MIL_inter} (see Appendix~\ref{appendix:usefullemmas}) with $x=q$, $a=M-m'-1$, and $b=M-1$, it is easy to see that the term $T_2$ of \eqref{eqn:g_increases_with_q} is also positive, thus $\frac{\partial \milh(q)}{\partial q}>0$. %\pk{June 1: Close the previous sentence. Then write like this: `Further, using  Lemma \ref{lemma:f_increases_with_x_WBU_MIL_inter} (see Appendix \ref{appendix:usefullemmas}) with $x=q$, $a=M-m'-1$, and $b=M-1$, it is easy to see that ....'. } we can prove that the term $T_2$ of \eqref{eqn:g_increases_with_q} is also positive, thus $\frac{\partial \milh(q)}{\partial q}>0$. 
    This implies $\milh$ is an increasing function in $q$, consequently, we can conclude $\milg$ is an increasing function in $q$. Thus, for any valid $q_1$ and $q_2$ such that 
    % integer pair of $(m', M)$ and 
    $0<q_1 < q_2 <1$ we have $\milg(q_1, m', M)< \milg(q_2, m', M)$. 
    
    Let $\milq=0.79587$. Since $\milg$ is increasing in $q$ for $q\in(0,1)$ for any $M>2$ and $m'\in[1:M-2]$, to complete the proof of Theorem~\ref{thm:new_WBU_MIL_thm}, it is thus sufficient to show that $\milg(\milq,m',M)<0$ for each $M>2$ and every choice of $m'\in[1:M-2]$. The behavior of $\milg(\milq,m',M)$ is depicted in Figure~\ref{fig:mil_g} for various values of $M$ and for $m'\in[0,M-1]$.

    % Let $\milq=0.7828$. 
    % \pk{June 1: WHY DON'T YOU USE $q_0$ EVERYWHERE BELOW, BY ABOVE LAST LINE? THE USAGE OF $q$ MAKES IT VERY CONFUSION, IT MAKES IT SEEM AS THOUGH THERE IS ANOTHER VARIABLE WHICH IS PRECISELY WHAT WE WANT TO AVOID.} %Now, we
    The proof completes in two parts:
    first for the range $3\le M< 14$, and then for the remaining case $M\ge14$.  
    
    \textbf{Case 1: $M\in[3:14]$:} 
    %Therefore,
    In this case, it is sufficient to show that the value of $\milg(\milq, m', M)$ is negative, for all valid pairs of $(m', M)$ for $M\in[3:14]$. We show this in Table~\ref{tab:MIL_M_less_than_14}. 
    % Thus, for $M<15$ we confirm that $\milg$ is strictly negative.

    \begin{table}[h!]
    \centering
    \scriptsize
    \renewcommand{\arraystretch}{1.18}
    \caption{Values of $\milg(\milq,\, m',\, M)$ (here $\milq=0.795877$) for all valid $m' \in [1:M-2]$, and  $M \in [3:13]$.
    All entries are negative, confirming $D_{m'} < 0$ throughout.}
    \label{tab:MIL_M_less_than_14}
    \setlength{\tabcolsep}{5pt}
    \begin{tabular}{cc}

    \begin{tabular}{ccS[table-format=1.5]S[table-format=1.5]S[table-format=-1.5]}
    \toprule
    $M$ & $m'$ & {$\frac{\milq^{m'+1}-\milq^M}{1-\milq^{M-1}}$} & {$1-\frac{\log(m'+1)}{\log M}$} & {$\milg$} \\
    \midrule
    3 & 1 & 0.35271 & 0.36907 & -0.01636 \\
    \midrule
    4 & 1 & 0.46826 & 0.50000 & -0.03174 \\
    4 & 2 & 0.20752 & 0.20752 & {$-0.00000$} \\
    \midrule
    5 & 1 & 0.52456 & 0.56932 & -0.04476 \\
    5 & 2 & 0.30863 & 0.31739 & -0.00876 \\
    5 & 3 & 0.13678 & 0.13865 & -0.00187 \\
    \midrule
    6 & 1 & 0.55721 & 0.61315 & -0.05594 \\
    6 & 2 & 0.36726 & 0.38685 & -0.01960 \\
    6 & 3 & 0.21608 & 0.22629 & -0.01022 \\
    6 & 4 & 0.09576 & 0.10176 & -0.00600 \\
    \midrule
    7 & 1 & 0.57807 & 0.64379 & -0.06573 \\
    7 & 2 & 0.40471 & 0.43542 & -0.03071 \\
    7 & 3 & 0.26675 & 0.28759 & -0.02084 \\
    7 & 4 & 0.15694 & 0.17291 & -0.01597 \\
    7 & 5 & 0.06955 & 0.07922 & -0.00967 \\
    \midrule
    8 & 1 & 0.59223 & 0.66667 & -0.07444 \\
    8 & 2 & 0.43015 & 0.47168 & -0.04153 \\
    8 & 3 & 0.30116 & 0.33333 & -0.03218 \\
    8 & 4 & 0.19849 & 0.22602 & -0.02753 \\
    8 & 5 & 0.11678 & 0.13835 & -0.02156 \\
    8 & 6 & 0.05176 & 0.06422 & -0.01246 \\
    \midrule
    9 & 1 & 0.60225 & 0.68454 & -0.08228 \\
    9 & 2 & 0.44815 & 0.50000 & -0.05185 \\
    9 & 3 & 0.32550 & 0.36907 & -0.04357 \\
    9 & 4 & 0.22789 & 0.26751 & -0.03962 \\
    9 & 5 & 0.15020 & 0.18454 & -0.03433 \\
    9 & 6 & 0.08837 & 0.11438 & -0.02601 \\
    9 & 7 & 0.03916 & 0.05361 & -0.01444 \\
    \midrule
    10 & 1 & 0.60955 & 0.69897 & -0.08942 \\
    10 & 2 & 0.46125 & 0.52288 & -0.06163 \\
    10 & 3 & 0.34323 & 0.39794 & -0.05471 \\
    10 & 4 & 0.24930 & 0.30103 & -0.05173 \\
    10 & 5 & 0.17454 & 0.22185 & -0.04731 \\
    10 & 6 & 0.11504 & 0.15490 & -0.03987 \\
    10 & 7 & 0.06768 & 0.09691 & -0.02923 \\
    10 & 8 & 0.02999 & 0.04576 & -0.01576 \\

    \bottomrule
    \end{tabular}
    &

    \begin{tabular}{ccS[table-format=1.5]S[table-format=1.5]S[table-format=-1.5]}
    \toprule
    $M$ & $m'$ & {$\frac{\milq^{m'+1}-\milq^M}{1-\milq^{M-1}}$} & {$1-\frac{\log(m'+1)}{\log M}$} & {$\milg$} \\
    \midrule

    11 & 1 & 0.61497 & 0.71094 & -0.09596 \\
    11 & 2 & 0.47100 & 0.54184 & -0.07085 \\
    11 & 3 & 0.35641 & 0.42187 & -0.06546 \\
    11 & 4 & 0.26521 & 0.32881 & -0.06360 \\
    11 & 5 & 0.19263 & 0.25278 & -0.06015 \\
    11 & 6 & 0.13486 & 0.18849 & -0.05363 \\
    11 & 7 & 0.08889 & 0.13281 & -0.04392 \\
    11 & 8 & 0.05230 & 0.08369 & -0.03139 \\
    11 & 9 & 0.02318 & 0.03975 & -0.01657 \\
    \midrule
    12 & 1 & 0.61907 & 0.72106 & -0.10199 \\
    12 & 2 & 0.47836 & 0.55789 & -0.07953 \\
    12 & 3 & 0.36636 & 0.44211 & -0.07575 \\
    12 & 4 & 0.27723 & 0.35231 & -0.07508 \\
    12 & 5 & 0.20630 & 0.27894 & -0.07265 \\
    12 & 6 & 0.14984 & 0.21691 & -0.06707 \\
    12 & 7 & 0.10490 & 0.16317 & -0.05827 \\
    12 & 8 & 0.06914 & 0.11577 & -0.04663 \\
    12 & 9 & 0.04068 & 0.07337 & -0.03269 \\
    12 & 10 & 0.01803 & 0.03502 & -0.01699 \\
    \midrule
    13 & 1 & 0.62220 & 0.72976 & -0.10756 \\
    13 & 2 & 0.48398 & 0.57168 & -0.08770 \\
    13 & 3 & 0.37397 & 0.45952 & -0.08555 \\
    13 & 4 & 0.28642 & 0.37253 & -0.08611 \\
    13 & 5 & 0.21674 & 0.30144 & -0.08471 \\
    13 & 6 & 0.16128 & 0.24135 & -0.08007 \\
    13 & 7 & 0.11714 & 0.18929 & -0.07215 \\
    13 & 8 & 0.08201 & 0.14337 & -0.06135 \\
    13 & 9 & 0.05405 & 0.10229 & -0.04823 \\
    13 & 10 & 0.03180 & 0.06513 & -0.03333 \\
    13 & 11 & 0.01409 & 0.03121 & -0.01711 \\
    % \midrule

    % 14 & 1 &0.60542 &0.73735 &-0.13193 \\
    % 14 & 2 &0.46658 &0.58371 &-0.11713 \\
    % 14 & 3 &0.35788 &0.47470 &-0.11682 \\
    % 14 & 4 &0.27280 &0.39015 &-0.11735 \\
    % 14 & 5 &0.20620 &0.32106 &-0.11486 \\
    % 14 & 6 &0.15406 &0.26265 &-0.10859 \\
    % 14 & 7 &0.11325 &0.21205 &-0.09880 \\
    % 14 & 8 &0.08130 &0.16742 &-0.08612 \\
    % 14 & 9 &0.05629 &0.12750 &-0.07121 \\
    % 14 &10 &0.03671 &0.09138 &-0.05467 \\
    % 14 &11 &0.02139 &0.05841 &-0.03702 \\
    % 14 &12 &0.00939 &0.02808 &-0.01869 \\
    \bottomrule
    \end{tabular}
    \end{tabular}
    \end{table}

    \textbf{Case 2: $M\ge14$:} We now consider the remaining case $M\ge14$. % and $q\le\milq$.
    First, we analyze the behavior of $\milg$, assuming that $\milg$ is a continuous function of $m'$ over the interval $(-1, M-1]$.
    The partial derivative of $\milg$ with respect to $m'$ is given by the following

    \begin{align}
        \frac{\partial \milg(\milq, m', M)}{\partial m'} &=\frac{\milq^{m'+1}\ln(\milq)}{1-q^{M-1}} + \frac{1}{(m'+1)\ln(M)}\nonumber\\
        &=-\frac{\milq^{m'+1}\ln(1/\milq)}{1-\milq^{M-1}} + \frac{1}{(m'+1)\ln(M)}\nonumber\\
        &=\frac{\ln(1/\milq)}{(1-\milq^{M-1})(m'+1)}\left(\frac{1-\milq^{M-1}}{\ln(M)\ln(1/\milq)}-(m'+1)\milq^{m'+1}\right).\label{eqn:MIL_partial_deriv_in_m}
    \end{align}

    The sign of \eqref{eqn:MIL_partial_deriv_in_m} depends on the term in the parentheses. We define two functions $\milC(\milq, M)\triangleq\frac{1-\milq^{M-1}}{\ln(M)\ln(1/\milq)}$ and $\milf(m')\triangleq(m'+1)\milq^{m'+1}$. 
   %%%
   \begin{observation}   
   \label{obs:milcmilf}
    We have the following observations to be true: 
    \begin{enumerate}[label=\alph*)]
            \item\label{obs:MIL_C} $\milC(\milq, M)>0$ as $M\ge14$, and further $\milC(\milq, M)$ is a constant with respect to $m'$.
            \item  By definition, the function $\milf(m')$ is positive for $m'>-1$. Further, we have $\milf(-1^+)=0$, and also $\lim_{m'\rightarrow\infty}\milf(m')=0$.  Since $\frac{\partial \milf(m')}{\partial m'}=\milq^{m'+1}(1+(m'+1)\ln(\milq))$, the function $\milf$ has a unique maximum value in the range $(-1,\infty)$ only at $m_0=\frac{1}{\ln(1/\milq)}-1=3.3798$, which is $\milf(m_0)=1.6112$. Thus, $\milf$ is increasing in the regime $m'\in[-1,m_0]$ and decreasing in $m'\in[m_0,\infty)$.
        \end{enumerate}
    \end{observation}
    %%%
    We now analyze two sub-cases based on the relationship between $\milC(\milq,M)$ and $\milf(m_0)$.
    \textbf{Case 2(a): $\milC(\milq,M)\ge\milf(m_0):$} If $\milC(\milq,M)\ge\milf(m_0)>\milf(m')$ (by Observation \ref{obs:milcmilf}, this last inequality holds), then from \eqref{eqn:MIL_partial_deriv_in_m}, we see that $\frac{\partial \milg}{\partial m'}\geq 0$ for $m'=m_0$ and $\frac{\partial \milg}{\partial m'}>0$ for remaining range $m'\in(-1,M-1]\backslash\{m_0\}$. Hence, $\milg$ is monotonically non-decreasing in $m'$, and further we can obtain $\lim_{m'\rightarrow-1^+}\milg(\milq, m', M)=-\infty$ and $\milg(\milq,M-1,M)=0$. Combining all the above properties, it follows that $\milg$ must be negative for all $m'\in[1:M-2]$, concluding the proof in this sub-case.

   % \begin{figure}
   %     \centering
   %     \includegraphics[width=0.9\linewidth]{New_arxiv_doc/MIL_C_and_f_curves.png}
   %     \caption{Plot for functions $\milC$ and $\milf$ with respect to $m'$.}
   %     \label{fig:placeholder}
   % \end{figure}
   \begin{figure}[htbp]
        \centering
        
        % First subfigure
        \subfloat[\label{fig:mil_c_f}]{%
            \includegraphics[width=0.49\linewidth]{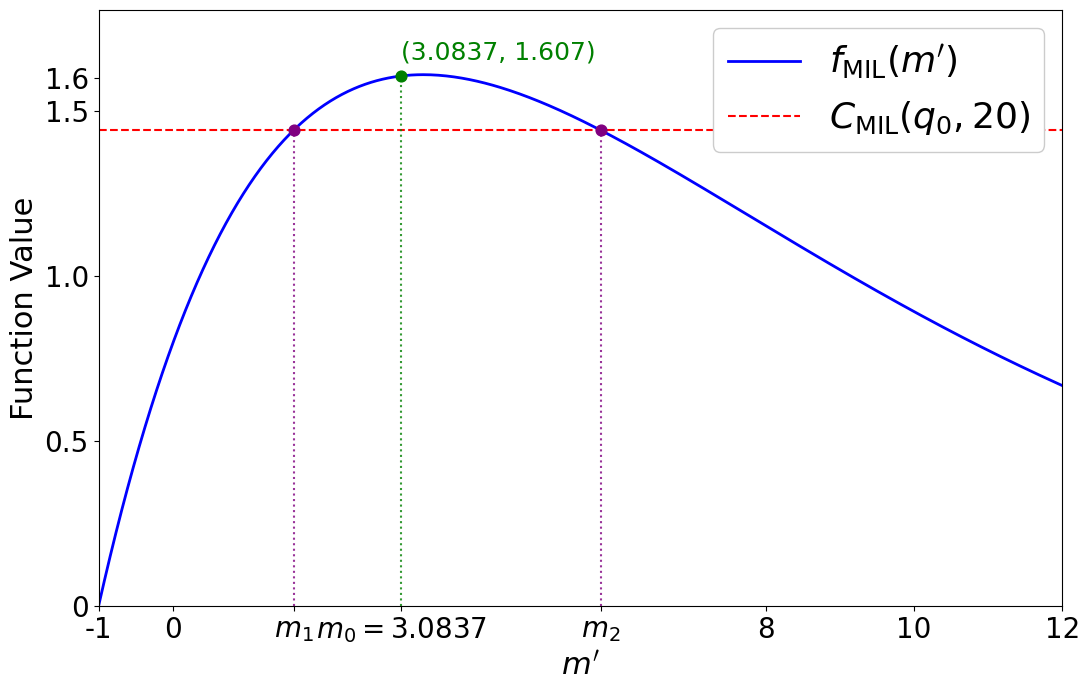}%
        }
        \hfill % Adds horizontal space between the two figures
        % Second subfigure
        \subfloat[\label{fig:mil_g}]{%
            \includegraphics[width=0.49\linewidth]{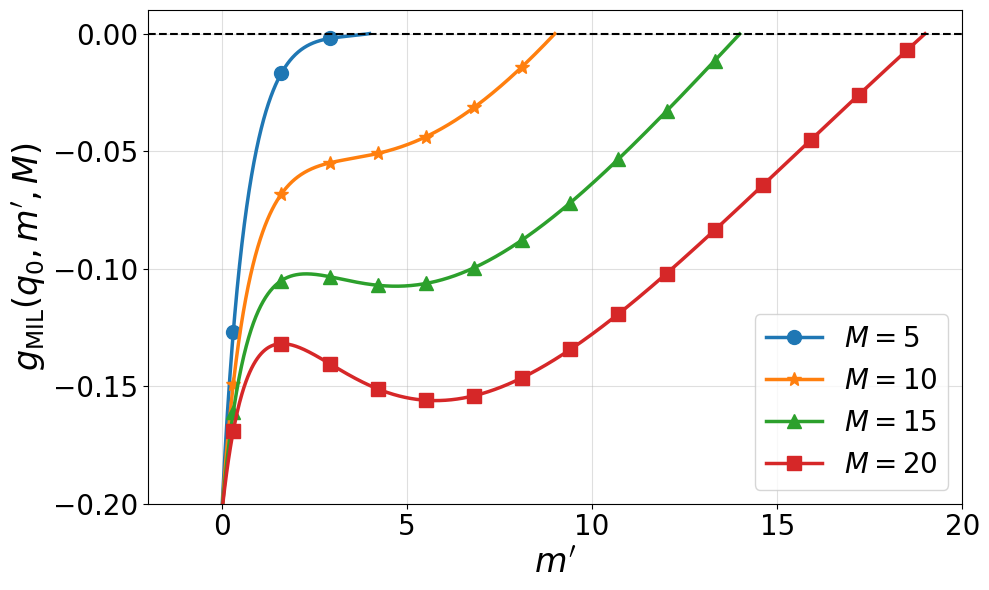}%
        }
        
        % Main figure caption
        \caption{Figure~\ref{fig:mil_c_f} shows the behavior of $\milC(\milq,M)$ and $\milf(m')$ as functions of $m'\ge0$, calculated at $\milq=0.795877$ with $M=20$.  %\pk{June 2026: Fig \ref{fig:mil_c_f}  is only for  $M=14$? Then mention the same.}\chandan{Done}. 
        The behavior of $\maxlg(\milq,m',M)$ on $m'\in[0,M-1]$ across varying system scales is examined in Figure~\ref{fig:mil_g}, where $M\in\{5,10,15,20\}$ at the same fixed $\milq=0.795877$. %The contiguous part of the curve represents the region $-1<m'\le M-1$ and the dotted curve represents $m'>M-1$. %\pk{Text/values in the axes for both the figures, should be bigger. I can hardly read them. In Fig \ref{fig:mil_g}, I dont think values beyond $M-1$ should be shown on the plot. It is very confusing. Remove dotted lines portion fully. Try to replace the different-$M$-curves by different patterns and also different colours}
        }
        \label{fig:mil_combined}
    \end{figure}
    
   \textbf{Case 2(b): $\milC(q, M)< \milf(m_0):$}  
   First, we make the following observation, which holds in this sub-case. 
   \begin{observation}
       \label{obs:structureofgmilatm1m2_Cmil<fmil-sub-case}
       If $\milC(q, M)<\milf(m_0)$, as depicted in Figure~\ref{fig:mil_c_f}, then by the structure of $\milf(m')$ as in Observation \ref{obs:milcmilf}, the equation $\milf(m')=\milC(q, M)$ (as an equation in $m'$) has exactly two roots, say $m_1$ and $m_2$, which satisfy $m_1<m_0<m_2$. It must be the case that $m_1< m_0 =3.3798$, and $M\ge14$, therefore we have that $m_1<M-1$. Further, as $\milf$ has a unique maximum at $m_0$, we see that $\milC-\milf$ (and thus $\frac{\partial \milg}{\partial m'}$, by \eqref{eqn:MIL_partial_deriv_in_m}) is positive in $m'\in[-1,m_1)\cup(m_2,\infty)$ and negative for $m'\in(m_1,m_2)$, in this sub-case. % \pk{June 2026: check these ranges for the behaviour of the derivative, carefully}. 
       Further, $\milg$ has a local maxima at $m'=m_1$, where $m_1\in[-1,m_0)$, and a local minima at $m'=m_2$, where $m_2\in(m_0,M-1]$. 
   \end{observation}

   To show $\milg$ is negative in this sub-case, for all integers $m'\in[1:M-1]$, we prove this in two parts: first, for the range $0<m'<m_0=3.3798$, and second, for the range $m_0\leq m'<M-1$. 

   \textbf{Case 2(b)(I) $0<m'<m_0$:} To show $\milg(m_1,M)$ (we abuse the notation $\milg(\milq,m_1,M)$ slightly, as  $\milq$ is a constant) is negative (for any $M\geq 14$) in this case, first we define an auxiliary function denoted as $\milgprime$ and show that it is an upper bound on $\milg$. Then, we show that $\milgprime$ is negative in the range $0<m'<m_0$ by a careful analysis of its maximum value in this range, which completes the proof of Case 2(b)(I). We now proceed with these steps.

    \textbf{Step (i)} Consider the function $\milgprime(m',M)\triangleq \milq^{m'+1}+\frac{\ln(m'+1)}{\ln(M)}-1$. We now prove that $\milgprime(m',M)$ is an upper bound on $\milg(m',M)$ for all $m'>0$.  
%%%%
    \begin{align}
        \milg(m',M)-\milgprime(m',M)&=\frac{\milq^{m'+1}-\milq^M}{1-\milq^{M-1}}-\milq^{m'+1}\nonumber\\
        &=\frac{\milq^{m'+1}-\milq^M-\milq^{m'+1}(1-\milq^{M-1})}{1-\milq^{M-1}}\nonumber\\
        &=\frac{\milq^M(\milq^{m'}-1)}{1-\milq^{M-1}}\label{eqn:MIL_aprox_g_g'}
    \end{align}

    Plugging in $\milq,m_0$, we see that the expression $\milq^{m'}-1<0$ for $m'\in(0,m_0)$. Hence from \eqref{eqn:MIL_aprox_g_g'} we have $\milg(m',M)<\milgprime(m',M)$. %Therefore, to prove $\milg$ is negative in the desired domain, it is enough to prove $\milgprime(q, m', M)<0$ when $m'=m_1$ at local maxima of $\milg$. %\pk{Again $m,m'$ confusion in prev sentence. Till here}

    \textbf{Step (ii)} Now, we show that $\milgprime$ is negative for $m'\in(0,m_0)$ for any $M\geq 14$. We first observe that as $\milgprime$ is clearly decreasing with $M$, it is sufficient to prove that $\milgprime(m',14)$ is negative as $m'\in(0,m_0)$. 
    %Furthermore, since $\frac{\partial \milgprime}{\partial M}=-\frac{\ln(m'+1)}{M\ln(M)}$ is negative whereas $\frac{\partial \milgprime}{\partial q}=(m'+1)q^{m'}$ is positive for $q\in(0,\milq)$, $M\ge14$. Therefore, $\milgprime$ is increasing in $q$ and decreasing in $M$. \pk{June 5: no need for all this, why we need behaviour wrt $M$ and $\milq$?}
    % \chandan{Partial derivative of $g_{MIL}'$ with respect to $m$.} 
    % We will analyze t
   
    To show that $\milgprime(m',14)$ is negative as $m'\in(0,m_0)$, we first show that the derivative of $\milgprime$ with respect to $m'$ is positive in this range, and thus it is sufficient to show $\milgprime(m_0,14)$ is negative. 
    
    The derivative of $\milgprime(m',14)$ with respect to $m'$ is as follows: %for $m'<m_1\le3.0837$. %4658$.
    \begin{align}
        \frac{\partial \milgprime}{\partial m'} &= \milq^{m'+1}\ln(\milq)+\frac{1}{(m'+1)\ln(14)}.
        % \nonumber\\
        % &= -\frac{\ln(1/\milq)}{m'+1}\left((m'+1)\milq^{m'+1}-\frac{1}{\ln(14)\ln(1/\milq)}\right)
        \label{eqn:MIL_partial_g'_w.r.t_m}
    \end{align}

Taking the second derivative, we get, 
\begin{align}
        \frac{\partial^2 \milgprime}{\partial (m')^2} &= \milq^{m'+1}(\ln(\milq))^2-\frac{1}{(m'+1)^2\ln(14)}\nonumber\\
        &= \frac{(\ln(\milq))^2}{(m'+1)^2}\left(\milq^{m'+1}(m'+1)^2-\frac{1}{\ln(\milq))^2\ln(14)}\right).
        \label{eqn:MIL_doublepartial_g'_w.r.t_m}
    \end{align}
The sign of \eqref{eqn:MIL_doublepartial_g'_w.r.t_m} depends on the sign of the term within the parentheses. Note that the term $\milk(m')\triangleq\milq^{m'+1}(m'+1)^2$ has derivative 
\begin{align*}
    \frac{\partial \milk(m')}{\partial m'}=\milq^{m'+1}(m'+1)\left(2+(m'+1)\ln(\milq)\right). 
\end{align*}
The minimum value of $2-(m'+1)\ln(1/\milq)$ in the range $m'\in(0,m_0)$ occurs at $m_0=\frac{1}{\ln(1/\milq)}-1$, at which point we have $2-(m_0+1)\ln(1/\milq)>2-(\frac{1}{\ln(1/\milq)})(\ln(1/\milq))=1>0$. Thus, $\milk(m')$ is increasing in $m'\in(0,m_0)$. So the largest value of $\milk(m')$ occurs at $m'=m_0$, which is evaluated as $\milk(m_0)=\milq^{m_0+1}(m_0+1)^2<6.2<(\ln(\milq))^2\ln(14).$ Therefore, from \eqref{eqn:MIL_doublepartial_g'_w.r.t_m}, we see that $\frac{\partial^2 \milgprime}{\partial (m')^2}$ is negative for $m'\in(0,m_0)$. Thus, $\frac{\partial \milgprime}{\partial m'}$ is strictly decreasing in $m'$ in this range. Thus, the minimum value of $\frac{\partial \milgprime}{\partial m'}$ occurs at $m'=m_0$, at which point we have, from \eqref{eqn:MIL_partial_g'_w.r.t_m}. 
%%%
$$\frac{\partial \milgprime}{\partial m'}\bigg|_{m'=m_0}= 0.0025>0$$. 
%%%
Thus, $\milgprime(m',14)$ is increasing in $m'\in(0,m_0)$. Thus, its largest value is attained at $m_0$, which we is calculated as follows.

\begin{align}
        \milgprime(m_0,14)&=\milq^{m_0+1}+\frac{\ln(m_0+1)}{\ln(14)}-1\nonumber\\
        &=0.79587^{4.3798}+\frac{\ln(4.3798)}{\ln(14)}-1\nonumber\\
        &=0.3679+0.5597-1\nonumber\\
        &=-0.0724<0.
    \end{align}
As, $\milgprime(m',M)<\milgprime(m',14)<\milgprime(m_0,14)$ for all $m'\in(0,m_0)$ and for all $M\geq 14$, we thus have $\milgprime(m',M)$ is negative in $m'\in(0,m_0)$, which by \eqref{eqn:MIL_aprox_g_g'} means that  $\milg(m',M)$ is negative in $m'\in(0,m_0)$. This concludes the present sub-case.

    \textbf{Case 2(b)(II) $m_0\leq m'<M-1$:} 
    By Observation \ref{obs:structureofgmilatm1m2_Cmil<fmil-sub-case}, we know $\milg$ attains a maximum value at $m'=m_1$. By Case 2(b)(I), we know that $\milg(m_1)<0$, as $m_1<m_0$. Further,  Observation \ref{obs:structureofgmilatm1m2_Cmil<fmil-sub-case} guarantees that $\milg$ is decreasing with $m'$ in $m'\in(m_1,m_2)$, attaining a minimum at $m'=m_2$. Thus, $\milg(m')<\milg(m_1)$ for any $m'\in(m_1,m_2]$. Note that $m_0\in(m_1,m_2)$, thus $\milg(m')<\milg(m_1)<0$ for $m'\in[m_0,m_2]$. 
    
    Now, in the interval $(m_2,M-1]$, $\milg$ is increasing as per Observation \ref{obs:structureofgmilatm1m2_Cmil<fmil-sub-case}. However, by its definition in \eqref{eqn:MIL_g_function}, we see that for each $M>2$ we have $\milg(M-1, M)=0$. Combining these two observations, we see that $\milg(m')<0$ for $m'\in(m_2, M-2]$. This concludes the proof that $\milg$ is negative for all $m'\in[1, M-2]$, for all $M>2$, and for all $q\leq q_0=0.79587$. Thus, the proof of Theorem \ref{thm:new_WBU_MIL_thm} is complete.

\section{Proof of Theorem \ref{thm:new_WBU_MaxL_thm}: Class-wise optimality of Theorem \ref{thm:WBU_MaxL} under $\mathsf{MaxL}$}
\label{sec:proofofThmWBUMaxL}
% \section{Proof of Theorem \ref{thm:new_WBU_MaxL_thm}}\label{sec:proofofThmWBUMaxL}
% \subsection{Optimal rate of $\wbu$ scheme and $\wcsj$ scheme under $\maxleak$}
% In the case of $\wbu$ under $\maxleak$, the derivation follows the same steps as in Subsection \ref{subsec:WSJ_optimal_rate_for_MaxL}, with the modification that $N$ is replaced by $N/K$ and the constraint on the threshold of the storage rate $K/N$. Accordingly, the equations in the theorem below are adjusted.

% \begin{theorem}
%     \label{thm:new_WBU_MaxL_thm}
%     For $\frac{K}{N} \leq 0.68$, and any $\rho\in\left[0,\log (1 + \frac{K(M-1)}{N})\right]$ \pk{Jan 17: the upper limit creates a problem (see $P_{M'}(0)$), must be avoided. But more series problems exist}, the maximal rate of the $\wbu$ scheme under a $\mathsf{MaxL}$ constraint $\maxleak \leq \rho$ is
    
%     {\small$$R^*_{WBU} = \left(1+ \left(1-\frac{N(2^\rho - 1)}{K(M-1)}\right) \left(\frac{K}{N} + \cdots + \left(\frac{K}{N}\right)^{M-1} \right)\right)^{-1}$$}
    
%     if and only if the distribution $\Prob_{M'}$ is given as \begin{itemize}
%         \item $\Prob_{M'}(0)=\frac{(2^{\rho}-1) N}{K(M-1)}=1-\Prob_{M'}(M-1)$ and $\Prob_{M'}(m')=0$ for all $m'\in[1:M-2]$.
%     \end{itemize}
% \end{theorem}

In this section, we provide the proof of Theorem \ref{thm:new_WBU_MaxL_thm}. As established in Observation \ref{obs:for_M_equels_2_MaxL}, when $M=2$ the distribution of $M'$ is uniquely determined by the $\mathsf{MaxL}$ leakage constraint $\rho$, and consequently the corresponding rate-privacy trade-off is fixed and class-wise optimal. Therefore, in the remainder of the analysis, we restrict attention to the case $M>2$. Since the maximal achievable rate $R_{WBU,\mathsf{MaxL}}^{*}$ is increasing monotonically in $\rho$, it suffices to prove the result for all $\maxleak=\rho$. From Theorem \ref{thm:WBU_MaxL}, the distribution defined by $\Prob_{M'}(0)=\frac{N(2^{\rho}-1)}{K(M-1)}=1-\Prob_{M'}(M-1)$ and $\Prob_{M'}(m')=0$ for all $m'\in[1:M-2]$, achieves the rate $R_{WBU,\mathsf{MaxL}}(\Prob_{M'})=R_{WBU,\mathsf{MaxL}}^*$, while also attaining $\maxleak(\Prob_{M'})=\rho$. However, it remains to show that this is the unique distribution that attains the maximum rate under the same constraint $\maxleak=\rho$. This will complete the proof of the theorem.

We start the proof of this theorem along the same lines of Theorem \ref{thm:new_WBU_MIL_thm}, assuming that there exists some valid distribution $\mathbf{p}=(p(0),p(1),\hdots,p(M-1))$ for $\Prob_{M'}$, for which the value of $\maxleak$ as given in \eqref{eqn:WBU_MaxL_expression} is equal to $\rho$.
    \begin{align*}
        2^\rho=p(0)\left(\frac{K(M-1)+N}{N}\right)+ \sum_{m'=1}^{M-1}p(m')\left(\frac{M}{m'+1}\right).
    \end{align*}
    Fixing $b_0=\frac{K(M-1)+N}{N}$ and $b_{m'}=\frac{M}{m'+1}$ for $m'\in[1:M-1]$, we can write 
    \begin{align}
        p(0)b_0&=2^{\rho}-\sum_{m'=1}^{M-1}p(m')b_{m'}\nonumber\\
        &=2^{\rho}-1+1-\sum_{m'=1}^{M-1}p(m')b_{m'}\nonumber\\
        &=2^{\rho}-1+\sum_{m'=0}^{M-1}p(m')-\sum_{m'=1}^{M-1}p(m')b_{m'}\nonumber\\
        &=2^{\rho}-1+p(0)+\sum_{m'=1}^{M-1}p(m')(1-b_{m'})\label{eqn:MaxL_p0_distribution1}.
    \end{align}
    Subtract $p(0)$ from both sides of \eqref{eqn:MaxL_p0_distribution1} and divide both side by $b_0-1$ gives the following result: 
    \begin{equation}
        % p(0)b_0-p(0)=2^{\rho}-1+\sum_{m'=1}^{M-1}p(m')(1-b_{m'})\nonumber\\
        p(0)= \frac{2^\rho-1}{b_0-1}-\sum_{m'=1}^{M-2}p(m')\frac{b_{m'}-1}{b_0-1}.\label{eqn:p(0)_value_MaxL}
    \end{equation}
    Further, following ideas in Theorem \ref{thm:new_WBU_MIL_thm}, we define the term $$T(\mathbf{p})\triangleq\mathbb{E}_{M'}\left[\left(\frac{K}{N}\right)^{M'+1}\right]=\sum_{m'=0}^{M-1}p(m')c_{m'},$$ where $c_{m'}=\left(\frac{K}{N}\right)^{m'+1}$. We denote $T^*\triangleq c_{M-1} + \frac{2^{\rho}-1}{b_0-1}(c_0 - c_{M-1})=\left(\frac{K}{N}\right)^{M}+\frac{N(2^{\rho}-1)}{K(M-1)}\left(\frac{K}{N}-\left(\frac{K}{N}\right)^{M}\right)$. 
    Then, from $R^*_{WSJ,\mathsf{MaxL}}$ in Theorem \ref{thm:new_WBU_MaxL_thm}, we get the following. 
    \begin{align*}
        R^{*}_{WBU,\mathsf{MaxL}}
        &= \left(1+\left(1-\frac{N(2^{\rho}-1)}{K(M-1)}\right) \frac{\frac{K}{N}-\left(\frac{K}{N}\right)^{M}}{1-\frac{K}{N}}\right)^{-1} \nonumber \\
        &= \left(\frac{(1-\frac{K}{N}) + \left(1-\frac{N(2^{\rho}-1)}{K(M-1)}\right)\left(\frac{K}{N}-\left(\frac{K}{N}\right)^{M}\right)}{1-\frac{K}{N}}\right)^{-1} \nonumber \\
        &= \left(\frac{1 - \frac{K}{N} + \frac{K}{N} - \left(\frac{K}{N}\right)^{M} - \frac{N(2^{\rho}-1)}{K(M-1)}(\frac{K}{N}-\left(\frac{K}{N}\right)^{M})}{1-\frac{K}{N}}\right)^{-1} \nonumber \\
        &= \left(\frac{1 - \left[\left(\frac{K}{N}\right)^{M} + \frac{N(2^{\rho}-1)}{K(M-1)}\left(\frac{K}{N}-\left(\frac{K}{N}\right)^{M}\right)\right]}{1-\frac{K}{N}}\right)^{-1} \nonumber \\
        &= \left(\frac{1-T^*}{1-K/N}\right)^{-1} = \frac{1-K/N}{1-T^*}.
    \end{align*}
    
   Thus, the rate $R^{*}_{WBU,\mathsf{MaxL}}$ is achieved by fixing $T=T^*$. Thus, to prove Theorem \ref{thm:new_WBU_MaxL_thm}, it is sufficient is to show that, under the assumption that $K/N\leq 0.60199$, the only choice of the distribution $\bold{p}$ of $M'$ which gives $T(\bold{p})\geq T^*$ (and thus $R_{WBU,\mathsf{MaxL}}\geq R^{*}_{WBU,\mathsf{MaxL}}$) under $\maxleak= \rho$ is the distribution $p(0)=(N(2^\rho-1)/(K(M-1)=1-p(M-1).$ Again, we continue with arguments similar to the proof of Theorem \ref{thm:new_WBU_MIL_thm}.
   
   Now, using $p(0)$ from \eqref{eqn:p(0)_value_MaxL}, and using the fact that $p({M-1})=1-\sum_{m'=0}^{M-2}p(m')$, we get the following.
    \begin{align} 
        T(\mathbf{p}) &=\sum_{m'=0}^{M-2} p(m') c_{m'} +\left(1-\sum_{m=0}^{M-2}p(m')\right)c_{M-1}\nonumber\\
        &=c_{M-1}+p(0)\left(c_{0}-c_{M-1}\right) + \sum_{m'=1}^{M-2}p(m')\left(c_{m'}-c_{M-1}\right)\nonumber\\
        &=c_{M-1}+\left(\frac{2^\rho-1}{b_0-1}-\sum_{m'=1}^{M-2}p(m')\frac{b_{m'}-1}{b_0-1}\right)\left(c_{0}-c_{M-1}\right)\nonumber\\
        &~~~~~~+\sum_{m'=1}^{M-2}p(m')\left(c_{m'}-c_{M-1}\right)\nonumber\\
        &=\left[ c_{M-1} + \frac{2^\rho-1}{b_0-1}(c_0 - c_{M-1}) \right]\nonumber\\
        &~~~~+ \sum_{m'=1}^{M-2} p(m') 
        \left[ (c_{m'} - c_{M-1}) - \frac{b_{m'}-1}{b_0-1}(c_0 - c_{M-1}) \right]\nonumber
        \\
        &= T^* + \sum_{m'=1}^{M-2} p(m') \underbrace{\left[ (c_{m'} - c_{M-1}) - \frac{b_{m'}-1}{b_0-1}(c_0 - c_{M-1}) \right]}_{D_{m'}} \label{eq:T_expanded_MaxL},
    \end{align}
    where the last equality follows by definition of $T^*$. 

    Now, suppose we show that the coefficients $D_{m'}$ are strictly negative for $m'\in[1:M-2]$. Establishing this negativity implies that the objective function $T(\mathbf{p})$ is maximized strictly when $p(m')=0$ for all $m'\in[0:M-1]\setminus\{0,M-1\}$, thereby confining the support of the optimal $\Prob_{M'}$ to the boundary set $\{0,M-1\}$. Subject to this binary support, the constraint $\maxleak=\rho$ uniquely determines the probability distribution as $\Prob_{M'}(0)=1-\Prob_{M'}(M-1)=\frac{N(2^{\rho}-1)}{K(M-1)}$. Thus, this would result in the proof of Theorem \ref{thm:new_WBU_MaxL_thm}. 
    
    Hence, the rest of the proof focusses on showing $D_{m'}<0$, for $m'\in[1:M-2]$, when $K/N\leq 0.60199$. We consider the following equations.
    % Again, for the $\mathsf{MaxL}$ case, it suffices to show that $D_{m'}<0$ for all $m'\in\{1,2,\cdots, M-2\}$. The proof follows arguments similar to those used in Theorem \ref{thm:new_WBU_MIL_thm} to obtain $D_{m'}$. The remainder of the proof follows an approach different from that of Theorem \ref{thm:new_WBU_MIL_thm}.
    \begin{align}
        D_{m'} &= \left[ (c_{m'} - c_{M-1}) - \frac{b_{m'}-1}{b_0-1}(c_0 - c_{M-1}) \right] \notag \\
        &= \left(\left(\frac{K}{N}\right)^{m'+1}- \left(\frac{K}{N}\right)^M\right)-\left(
        \frac{\frac{M}{m'+1}-1}{\frac{K(M-1)+N}{N}-1}\right)\left(\frac{K}{N}-\left(\frac{K}{N}\right)^M\right)\nonumber\\
        &= \left(\left(\frac{K}{N}\right)^{m'+1}- \left(\frac{K}{N}\right)^M\right)-
        \frac{N}{K}\left(\frac{M-m'-1}{(m'+1)(M-1)}\right)\left(\frac{K}{N}-\left(\frac{K}{N}\right)^M\right)\nonumber\\
        &= \left(\left(\frac{K}{N}\right)^{m'+1}- \left(\frac{K}{N}\right)^M\right)-
        \left(1-\frac{m'}{m'+1}\frac{M}{M-1}\right)\left(1-\left(\frac{K}{N}\right)^{M-1}\right)\nonumber\\
        &= \left(1-\left(\frac{K}{N}\right)^{M-1}\right)\left(\frac{\left(\frac{K}{N}\right)^{m'+1}- \left(\frac{K}{N}\right)^M}{1-\left(\frac{K}{N}\right)^{M-1}}-\left(1-\frac{\frac{m'}{m'+1}}{\frac{M-1}{M}}\right)\right)\label{eqn:D_mprime_MaxL_expression}
        % &= \left(1-\frac{1}{N^{M-1}}\right)\left(\frac{\frac{1}{N^{m'+1}}- \frac{1}{N^M}}{1-\frac{1}{N^{M-1}}}-
        % \left(1-\left(\frac{m'}{m'+1}\right)\left(\frac{M}{M-1}\right)\right)\right)\nonumber\\
        % &= \left(1-\frac{1}{N^{M-1}}\right)\left(\frac{\frac{1}{N^{m'+1}}- \frac{1}{N^M}}{1-\frac{1}{N^{M-1}}}-
        % \left(1-\frac{\left(\frac{m'}{m'+1}\right)}{\left(\frac{M-1}{M}\right)}\right)\right)\nonumber
        % % &= \frac{1}{N^M(m'+1)}\Biggl((m'+1)\big({N^{M-m'-1}} - 1\big) \nonumber\\ &~~~~~~~- \frac{1}{N(M-1)}\left(\frac{M}{m'+1}-1\right)\big({N^{M-1}}-1\big)\Biggr)\label{eqn:for_Pm_maxl_criteria}
        % % &= \left(\frac{1}{N^{m+1}} - \frac{1}{N^M}\right) -
        % % \left(\frac{MN}{(m+1)(M+N-1)}\right)\left(\frac{1}{N}-\frac{1}{N^{M}}\right)\nonumber\\
    \end{align}

    Thus, from \eqref{eqn:D_mprime_MaxL_expression}, to show $D_{m'}<0$ for every $m'\in[1:M-2]$ and for $K/N\leq\maxlq$ where $\milq=0.60199$, it is sufficient to prove that %\pk{Apr 12: Pls use $g_{\mathsf{MaxL}}$ instead, similarly for $h$ also}\chandan{I defined macros for all necessary notations.}
    \begin{equation}
        \maxlg(q,m',M)\triangleq \frac{q^{m'+1}-q^{M}}{1-q^{M-1}}+\frac{m'}{m'+1}\left(\frac{M}{M-1}\right)-1\label{eqn:MaxL_g_definition}
    \end{equation}
    is strictly negative for $q\le \maxlq$, $M>2$, and $m'\in[1:M-2]$. 
    
    We begin by establishing the monotonicity of $\maxlg$ with respect to $q$. For fixed $m'$ and $M$, define
    \begin{equation}
        \maxlh(q)\triangleq q^{m'+1}\left(\frac{1-q^{M-m'-1}}{1-q^{M-1}}\right).\nonumber
    \end{equation}
    % Then
    % \[
    % g_{MaxL}(q,m',M)=h_{MaxL}(q)-\left(1-\frac{m'}{m'+1}\frac{M}{M-1}\right),
    % \]
    % and all terms except $h_{MaxL}(q)$ are independent of $q$. 
    % Thus, the monotonicity of $\maxlg$ in $q$ is equivalent to the monotonicity of $\maxlh$ in $q$. \pk{Apr 12: this following line changed a bit, pls check}
    Then $\maxlg(q,m',M)=\maxlh(q)+\frac{m'}{m'+1}\frac{M}{M-1}-1$, and except $\maxlh(q)$ the remaining terms in $\maxlg$ are independent of $q$, so the monotonicity of $\maxlg$ in $q$ is equivalent to the monotonicity of $\maxlh$ in $q$. Since $\maxlh(q)$ is identical to the function $\milh(q)$ defined in \eqref{eqn:hMIL_auxfunction_defn} in Section~\ref{sec:proofofThmWBUMIL}, the argument leading to \eqref{eqn:g_increases_with_q} applies verbatim here and gives
    \[
    \frac{\partial \maxlh(q)}{\partial q}>0,\qquad q\in(0,1).
    \]
    Therefore $\maxlh$, and hence $\maxlg$, is strictly increasing in $q$. Consequently, $q_1$ and $q_2$ such that $0<q_1<q_2<1$ we have 
    \begin{equation}
        \maxlg(q_1,m',M)<\maxlg(q_2,m',M)~~\forall~ m'>-1~\text{ and }~\forall~M\ge3.\label{eqn:MaxL_g_monotonicity_on_q}
    \end{equation}
    % \pk{June 3: Again you have to say here that, `To complete the proof, we only have to show $\maxlg(\milq,m',M)<0$ for $m'\in[1:M-2],~\forall M\geq 3$.' IN THE REST OF THE PROOF DO NOT USE $q$, USE ONLY $\milq$.}
    Since $\maxlg$ is increasing in $q$ for $q\in(0,1)$, any $M>2$ and $m'\in[1:M-2]$, to complete the proof of Theorem~\ref{thm:new_WBU_MaxL_thm}, it is sufficient to show that $\maxlg(\maxlq, m', M)<0$, $\forall M>2$ and  $m'\in[1:M-2]$. The behavior of $\maxlg(\maxlq, m', M)$ is shown in Figure~\ref{fig:maxl_g} for some $M$ and for $m'\in[0,M-1]$.

    % \pk{}
    Now, we write the proof in two parts: first for the range $3\le M\le 15$, and then for the remaining case $M>15$. The arguments are very similar to those in the proof of Theorem \ref{thm:new_WBU_MIL_thm}, repeating in structure almost as is, except for changes in the actual numerical values. For the sake of completeness, we provide the full proof with all details. 
    
    \textbf{Case 1:} %\chandan{Important update: Using the following technique, we can prove only for $0\le q\le0.601994$. For the remaining range $0.601994<q\le0.602118$, it is really difficult to prove. So I put a remark.}\pk{June 22: Since the threshold has changed to 0.601994, this comment of yours is not relevant anymore?}
    $M\in[3: 15]$: Using \eqref{eqn:MaxL_g_monotonicity_on_q} it is sufficient to check that $\maxlg(\maxlq,m',M)$ is negative for all valid $(m',M)$ with $M\in[3:15]$, as shown in Table~\ref{tab:MaxL_M_less_than_16}.% lists all such values, and every entry is strictly negative. Hence $\maxlg$ is strictly negative in this case.% for all $3\le M\le 15$.

    \begin{table}[h!]
    \centering
    \scriptsize
    \renewcommand{\arraystretch}{1.18}
    \caption{Values of $\maxlg(\maxlq,\, m',\, M)$ (here $\maxlq = 0.601994$) for all valid $m' \in [1:M-2]$, and  $M \in [3:15]$.
    All entries are negative, confirming $D_{m'} < 0$ throughout.}
    \label{tab:MaxL_M_less_than_16}
    \setlength{\tabcolsep}{5pt}
    \begin{tabular}{cc}

    \begin{tabular}{ccS[table-format=1.5]S[table-format=1.5]S[table-format=-1.5]}
    \toprule
    $M$ & $m'$ & {$\frac{\maxlq^{m'+1}-\maxlq^M}{1-\maxlq^{M-1}}$} & {$1-\frac{m'}{m'+1}\frac{M}{M-1}$} & {$\maxlg$} \\
    \midrule
    3 & 1 & 0.22622 & 0.25000 & -0.02378 \\
    \midrule
    4 & 1 & 0.29554 & 0.33333 & -0.03779 \\
    4 & 2 & 0.11106 & 0.11111 & -0.00005 \\
    \midrule
    5 & 1 & 0.32617 & 0.37500 & -0.04883 \\
    5 & 2 & 0.16013 & 0.16667 & -0.00654 \\
    5 & 3 & 0.06017 & 0.06250 & -0.00233 \\
    \midrule
    6 & 1 & 0.34183 & 0.40000 & -0.05817 \\
    6 & 2 & 0.18521 & 0.20000 & -0.01479 \\
    6 & 3 & 0.09093 & 0.10000 & -0.00907 \\
    6 & 4 & 0.03417 & 0.04000 & -0.00583 \\
    \midrule
    7 & 1 & 0.35042 & 0.41667 & -0.06624 \\
    7 & 2 & 0.19898 & 0.22222 & -0.02324 \\
    7 & 3 & 0.10781 & 0.12500 & -0.01719 \\
    7 & 4 & 0.05293 & 0.06667 & -0.01374 \\
    7 & 5 & 0.01989 & 0.02778 & -0.00789 \\
    \midrule
    8 & 1 & 0.35533 & 0.42857 & -0.07324 \\
    8 & 2 & 0.20684 & 0.23810 & -0.03126 \\
    8 & 3 & 0.11745 & 0.14286 & -0.02541 \\
    8 & 4 & 0.06364 & 0.08571 & -0.02208 \\
    8 & 5 & 0.03124 & 0.04762 & -0.01638 \\
    8 & 6 & 0.01174 & 0.02041 & -0.00867 \\
    \midrule
    9 & 1 & 0.35819 & 0.43750 & -0.07931 \\
    9 & 2 & 0.21142 & 0.25000 & -0.03858 \\
    9 & 3 & 0.12307 & 0.15625 & -0.03318 \\
    9 & 4 & 0.06988 & 0.10000 & -0.03012 \\
    9 & 5 & 0.03786 & 0.06250 & -0.02464 \\
    9 & 6 & 0.01859 & 0.03571 & -0.01713 \\
    9 & 7 & 0.00699 & 0.01562 & -0.00864 \\
    \midrule
    10 & 1 & 0.35988 & 0.44444 & -0.08456 \\
    10 & 2 & 0.21413 & 0.25926 & -0.04513 \\
    10 & 3 & 0.12639 & 0.16667 & -0.04027 \\
    10 & 4 & 0.07357 & 0.11111 & -0.03754 \\
    10 & 5 & 0.04178 & 0.07407 & -0.03230 \\
    10 & 6 & 0.02264 & 0.04762 & -0.02498 \\
    10 & 7 & 0.01111 & 0.02778 & -0.01667 \\
    10 & 8 & 0.00418 & 0.01235 & -0.00817 \\
    \midrule

    11 & 1 & 0.36089 & 0.45000 & -0.08911 \\
    11 & 2 & 0.21575 & 0.26667 & -0.05092 \\
    11 & 3 & 0.12837 & 0.17500 & -0.04663 \\
    11 & 4 & 0.07577 & 0.12000 & -0.04423 \\
    11 & 5 & 0.04411 & 0.08333 & -0.03923 \\
    11 & 6 & 0.02505 & 0.05714 & -0.03210 \\
    11 & 7 & 0.01357 & 0.03750 & -0.02393 \\
    11 & 8 & 0.00666 & 0.02222 & -0.01556 \\
    11 & 9 & 0.00250 & 0.01000 & -0.00750 \\
    \bottomrule
    \end{tabular}
    &

    \begin{tabular}{ccS[table-format=1.5]S[table-format=1.5]S[table-format=-1.5]}
    \toprule
    $M$ & $m'$ & {$\frac{\maxlq^{m'+1}-\maxlq^M}{1-\maxlq^{M-1}}$} & {$1-\frac{m'}{m'+1}\frac{M}{M-1}$} & {$\maxlg$} \\
    \midrule

    % 11 & 1 & 0.36089 & 0.45000 & -0.08911 \\
    % 11 & 2 & 0.21575 & 0.26667 & -0.05092 \\
    % 11 & 3 & 0.12837 & 0.17500 & -0.04663 \\
    % 11 & 4 & 0.07577 & 0.12000 & -0.04423 \\
    % 11 & 5 & 0.04411 & 0.08333 & -0.03923 \\
    % 11 & 6 & 0.02505 & 0.05714 & -0.03210 \\
    % 11 & 7 & 0.01357 & 0.03750 & -0.02393 \\
    % 11 & 8 & 0.00666 & 0.02222 & -0.01556 \\
    % 11 & 9 & 0.00250 & 0.01000 & -0.00750 \\
    % \midrule
    12 & 1 & 0.36149 & 0.45455 & -0.09305 \\
    12 & 2 & 0.21671 & 0.27273 & -0.05602 \\
    12 & 3 & 0.12955 & 0.18182 & -0.05226 \\
    12 & 4 & 0.07709 & 0.12727 & -0.05019 \\
    12 & 5 & 0.04550 & 0.09091 & -0.04541 \\
    12 & 6 & 0.02649 & 0.06494 & -0.03845 \\
    12 & 7 & 0.01504 & 0.04545 & -0.03042 \\
    12 & 8 & 0.00815 & 0.03030 & -0.02215 \\
    12 & 9 & 0.00400 & 0.01818 & -0.01418 \\
    12 & 10 & 0.00150 & 0.00826 & -0.00676 \\
    \midrule
    13 & 1 & 0.36185 & 0.45833 & -0.09648 \\
    13 & 2 & 0.21729 & 0.27778 & -0.06049 \\
    13 & 3 & 0.13026 & 0.18750 & -0.05724 \\
    13 & 4 & 0.07787 & 0.13333 & -0.05546 \\
    13 & 5 & 0.04634 & 0.09722 & -0.05089 \\
    13 & 6 & 0.02735 & 0.07143 & -0.04408 \\
    13 & 7 & 0.01592 & 0.05208 & -0.03616 \\
    13 & 8 & 0.00904 & 0.03704 & -0.02800 \\
    13 & 9 & 0.00490 & 0.02500 & -0.02010 \\
    13 & 10 & 0.00240 & 0.01515 & -0.01275 \\
    13 & 11 & 0.00090 & 0.00694 & -0.00604 \\
    \midrule
    14 & 1 & 0.36207 & 0.46154 & -0.09947 \\
    14 & 2 & 0.21764 & 0.28205 & -0.06441 \\
    14 & 3 & 0.13069 & 0.19231 & -0.06162 \\
    14 & 4 & 0.07835 & 0.13846 & -0.06011 \\
    14 & 5 & 0.04684 & 0.10256 & -0.05573 \\
    14 & 6 & 0.02787 & 0.07692 & -0.04905 \\
    14 & 7 & 0.01645 & 0.05769 & -0.04124 \\
    14 & 8 & 0.00958 & 0.04274 & -0.03316 \\
    14 & 9 & 0.00544 & 0.03077 & -0.02533 \\
    14 & 10 & 0.00295 & 0.02098 & -0.01803 \\
    14 & 11 & 0.00145 & 0.01282 & -0.01137 \\
    14 & 12 & 0.00054 & 0.00592 & -0.00537 \\
    \midrule
    15 & 1 & 0.36220 & 0.46429 & -0.10209 \\
    15 & 2 & 0.21785 & 0.28571 & -0.06787 \\
    15 & 3 & 0.13094 & 0.19643 & -0.06548 \\
    15 & 4 & 0.07863 & 0.14286 & -0.06423 \\
    15 & 5 & 0.04714 & 0.10714 & -0.06000 \\
    15 & 6 & 0.02818 & 0.08163 & -0.05345 \\
    15 & 7 & 0.01677 & 0.06250 & -0.04573 \\
    15 & 8 & 0.00990 & 0.04762 & -0.03772 \\
    15 & 9 & 0.00576 & 0.03571 & -0.02995 \\
    15 & 10 & 0.00327 & 0.02597 & -0.02270 \\
    15 & 11 & 0.00177 & 0.01786 & -0.01608 \\
    15 & 12 & 0.00087 & 0.01099 & -0.01012 \\
    15 & 13 & 0.00033 & 0.00510 & -0.00478 \\
    \bottomrule
    \end{tabular}
    \end{tabular}
    \end{table}

    \textbf{Case 2:} $M>15$: We now prove for the remaining case $M>15$. % and $q\le \maxlq$ \pk{$\milq$ ONLY!}.
    First, we analyze the behavior of $\maxlg$ with respect to $m'$, assuming that $\milg$ is a continuous function of the real variable $m'$ on the interval $(-1, M-1]$. %, and then restrict back to the integer points $m'=1, \ldots, M-2$ at the end, a similar approach we will apply here. 
    The partial derivative of $\maxlg$ with respect to $m'$ is
    \begin{align}
        \frac{\partial \maxlg(\maxlq, m', M)}{\partial m'}
        &= \frac{\maxlq^{m'+1}\ln(\maxlq)}{1-\maxlq^{M-1}}+\frac{1}{(m'+1)^2}\frac{M}{M-1}\nonumber\\
        &= -\frac{\maxlq^{m'+1}\ln(1/\maxlq)}{1-\maxlq^{M-1}}+\frac{1}{(m'+1)^2}\frac{M}{M-1}\nonumber\\
        &= \frac{\ln(1/\maxlq)}{(1-\maxlq^{M-1})(m'+1)^2}
        \left(\underbrace{\frac{M}{M-1}\frac{1-\maxlq^{M-1}}{\ln(1/\maxlq)}-(m'+1)^2\maxlq^{m'+1}}_{T_3}\right).\label{eqn:MaxL_partial_deriv_in_m}
    \end{align}
    
    % The sign of \eqref{eqn:MaxL_partial_deriv_in_m} is determined by the term inside the parentheses. 
    Note that the sign of \eqref{eqn:MaxL_partial_deriv_in_m} depends on the $T_3$ term. Now, we define two functions $\maxlC(q, M)\triangleq \frac{M}{M-1}\frac{1-q^{M-1}}{\ln(1/q)}$ and $\maxlf(m')\triangleq (m'+1)^2q^{m'+1}$.
    \begin{observation}
        \label{obs:maxlC_and_maxlf_properties}
        The following observations are true:
        \begin{enumerate}[label=\alph*)]
            \item $\maxlC(q_0, M)>0$ for all $M>15$, and further $\maxlC(q_0, M)$ is constant with respect to $m'$.
            \item By definition, the function $\maxlf(m')$ is positive for $m'>-1$. Further, we have $\maxlf(-1^+)=0$, and also $\lim_{m'\rightarrow\infty}\maxlf(m')=0$. Since $\frac{\partial \maxlf(m')}{\partial m'}=\maxlq^{m'+1}(m'+1)(2+(m'+1)\ln(\maxlq))$, the function $\maxlf$ has a unique maximum value in the range $m'\in[-1,\infty)$ only at $m_0=\frac{2}{\ln(1/\maxlq)}-1=2.9408$, which is $\maxlf(m_0)=2.102$. Thus, $\maxlf$ is increasing in the regime $m'\in[-1,m_0]$ and decreasing in $m'\in[m_0,\infty)$.
        \end{enumerate}
    \end{observation}
    %\chandan{13 April: $\milC$ is positive for all $q$ and $M$.}
    % and $\maxlf(m')\triangleq (m'+1)^2q^{m'+1}$, and note that for $0<q<\maxlq$, $M>2$ and $m'>-1$, we have $1-q^{M-1}>0$ and $\ln(1/q)>0$, thus in total it is easy to show that estimates of $\maxlC(q, M)$ and $\maxlf(m')$ are positive. The sign of \eqref{eqn:MaxL_partial_deriv_in_m} (or $\frac{\partial \maxlg}{\partial m'}$) is the same as the sign of $\maxlC(q, M)-\maxlf(m')$.
    
    % Now $\maxlf(m')>0$ for all $m'>-1$ and $q\in(0,1)$, also boundary values are $\lim_{m'\to -1^{+}}\maxlf(m')=0$ and $\lim_{m'\to\infty}\maxlf(m')=0$. Moreover, derivative of $\maxlf$ with respect to $m'$ is expressed as $\frac{\partial \maxlf(m')}{\partial m'}=q^{m'+1}(m'+1)\bigl(2+(m'+1)\ln q\bigr)$ and using this, the global maxima of $\maxlf$ occurs at $m'=\frac{2}{\ln(1/q)}-1$ which we labeled as $m_0$. %\chandan{May 25: Till here.}%\chandan{13 April: show double derivative is negative to claim maxima of $\maxlf$}. Hence $\maxlf$ has a unique global maximum at $m_0=\frac{2}{\ln(1/q)}-1$. For $q\le 0.602$, this gives $m_0<2.941$.
    % Now $\maxlf(m')>0$ for all $m'>-1$, and $\lim_{m'\to -1^{+}}\maxlf(m')=0, ~~\lim_{m'\to\infty}\maxlf(m')=0$. Moreover, $\frac{\partial \maxlf(m')}{\partial m'}
    % =q^{m'+1}(m'+1)\bigl(2+(m'+1)\ln q\bigr)$ \chandan{13 April: show double derivative is negative to claim maxima of $\maxlf$}. Hence $\maxlf$ has a unique global maximum at $m_0=\frac{2}{\ln(1/q)}-1$. For $q\le 0.602$, this gives $m_0<2.941$.
    
    We now analyze two sub-cases based on the relationship between $\maxlC(\maxlq, M)$ and $\maxlf(m_0)$.

    \textbf{Case 2(a): $\maxlC(\maxlq,M)\ge\maxlf(m_0):$} If $\maxlC(\maxlq,M)\ge\maxlf(m_0)>\maxlf(m')$ (by Observation \ref{obs:maxlC_and_maxlf_properties}, this last inequality holds), then from \eqref{eqn:MaxL_partial_deriv_in_m}, we see that $\frac{\partial \maxlg}{\partial m'}\geq 0$ for $m'=m_0$ and $\frac{\partial \maxlg}{\partial m'}>0$ for all $m'\in(-1,M-1]\backslash\{m_0\}$. Hence, $\maxlg$ is non-decreasing in $m'$, and further we can obtain $\lim_{m'\rightarrow-1^+}\maxlg(\maxlq, m', M)=-\infty$ and $\maxlg(\maxlq,M-1,M)=0$. Combining all the above properties, it follows that $\maxlg$ must be negative for all $m'\in[1:M-2]$, completes the proof in this sub-case.

    \textbf{Case 2(b): $\maxlC(\maxlq, M)< \maxlf(m_0):$} First, we make the following observations about the behavior of $\maxlC, \maxlf$ and $\maxlg$ with respect to $m'$ in this sub-case.
    \begin{observation}
        \label{obs:structureofgmaxlatm1m2_Cmaxl<fmaxl-sub-case}
        If $\maxlC(\maxlq, M)<\maxlf(m_0)$, as depicted in Figure~\ref{fig:maxl_c_f}, then by the structure of $\maxlf$ as in Observation~\ref{obs:maxlC_and_maxlf_properties} the equation $\maxlf(m')=\maxlC(\maxlq, M)$ (as an equation in $m'$) has exactly two roots, say $m_1$ and $m_2$, which satisfy $m_1<m_0<m_2$. Clearly, $m_1<m_0=2.9408$, and $M>15$, therefore we have that $m_1<M-1$. Further, as $\maxlf$ has a unique global maximum at $m_0$, we see that $\maxlC-\maxlf$ (and thus $\frac{\partial \maxlg}{\partial m'}$, by \eqref{eqn:MaxL_partial_deriv_in_m}) is positive in $m'\in[-1,m_1)\cup(m_2,\infty)$ and negative for $m'\in(m_1,m_2)$, in this sub-case. Further, $\maxlg$ has a local maxima at $m'=m_1$, where $m_1\in[-1,m_0)$, and a local minima at $m'=m_2$, where $m_2\in(m_0,M-1]$. %For any $m'\in(-1,\infty)\backslash\{m_0\}$ we have $\maxlC(\maxlq, M)=\maxlf(m_0)>\maxlf(m')$ which implies \eqref{eqn:MaxL_partial_deriv_in_m} is positive and at $m'=m_0$ the value of \eqref{eqn:MaxL_partial_deriv_in_m} is $0$. Overall $\maxlg$ is increasing in $m'$.
    \end{observation}

       \begin{figure}[htbp]
        \centering
        
        % First subfigure
        \subfloat[\label{fig:maxl_c_f}]{%
            \includegraphics[width=0.49\linewidth]{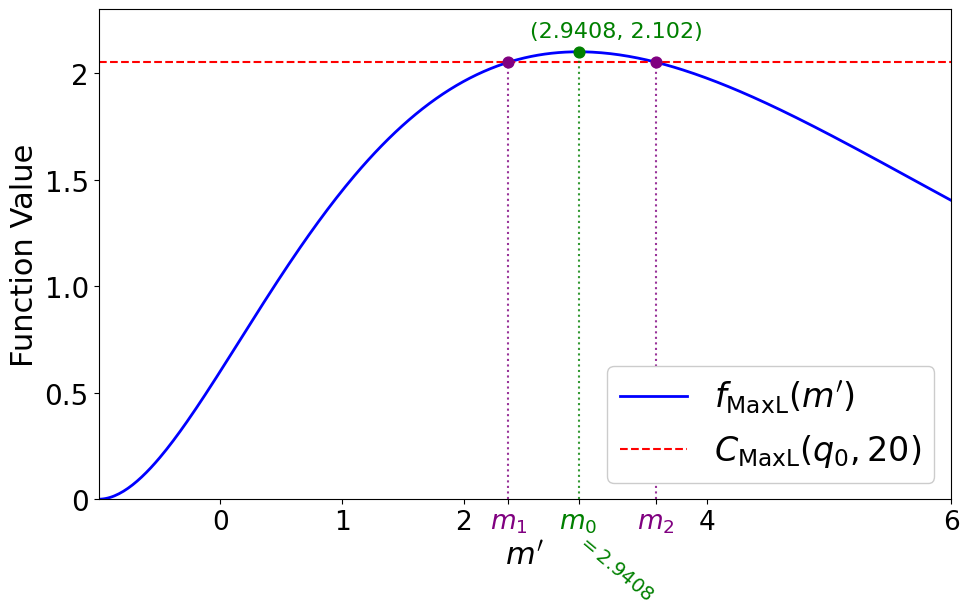}%
        }
        \hfill % Adds horizontal space between the two figures
        % Second subfigure
        \subfloat[\label{fig:maxl_g}]{%
            \includegraphics[width=0.49\linewidth]{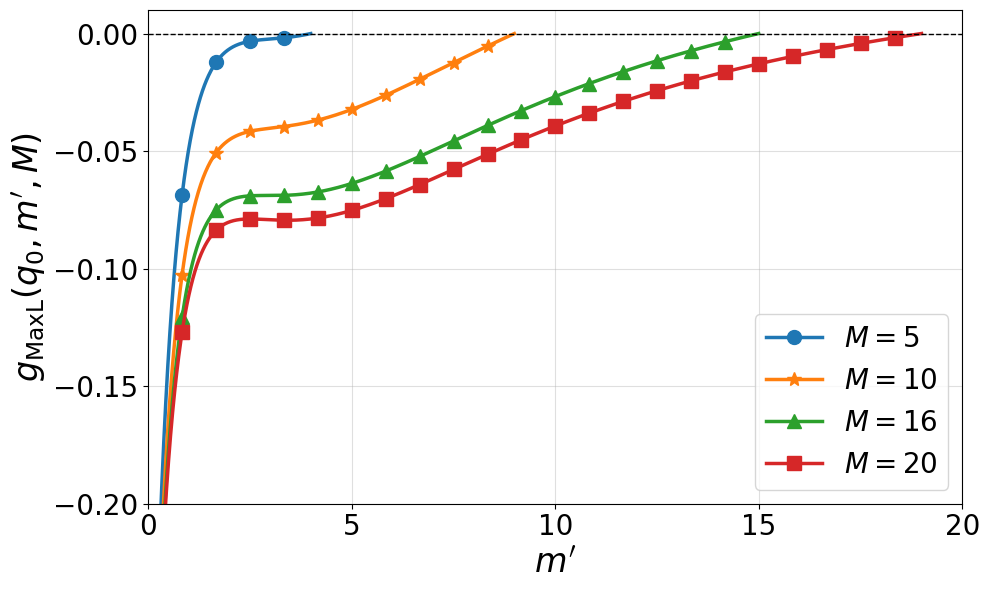}%
        }
        
        % Main figure caption
        \caption{Figure~\ref{fig:maxl_c_f} illustrates the behavior of $\maxlC$ and $\maxlf$ with respect to $m'\ge0$, computed at $\maxlq=0.601994$ with $M=20$. %as continuous functions in $m'\in[0,M-1]$.  %\pk{June 2026: Fig \ref{fig:maxl_c_f}  is only for  $M=14$? Then mention the same.}\chandan{Done}. 
        The behavior of $\maxlg(\maxlq,m',M)$ over $m'\in[0,M-1]$ is plotted in Figure~\ref{fig:maxl_g} for $\maxlq=0.601994$ and $M\in\{5,10,15,20\}$.%, \pk{for what values of $M$, $q$?}. \chandan{Done, please check.}
        }
        \label{fig:maxl_combined}
    \end{figure}

    % If $\maxlC(\maxlq, M)<\maxlf(m_0)$, then equation $\maxlf(m')=\maxlC(\maxlq, M)$ has exactly two roots, say $m_1$ and $m_2$, which satisfy $m_1<m_0<m_2$. If $\maxlC(q, M)=\maxlf(m_0)$ then this equation has exactly one root, and in that case, consider $m_1=m_0=m_2$. For any $m'\in(-1,\infty)\backslash\{m_0\}$ we have $\maxlC(q, M)=\maxlf(m_0)>\maxlf(m')$ which implies \eqref{eqn:MaxL_partial_deriv_in_m} is positive and at $m'=m_0$ the value of \eqref{eqn:MaxL_partial_deriv_in_m} is $0$. Overall $\maxlg$ is increasing in $m'$.

    % Since $m_1\le m_0=2.94$, and $M>15$, therefore we have that $m_1<M-1$.
    To show $\maxlg$ is negative in this sub-case, for all $m'\in[-1:M-1]$, we prove this in two parts: first, for the range $m'\in(0,m_0)$, and second, for the range $m'\in(m_0,M-1]$. %, it is sufficient to show that the local maximum $\maxlg(q, m_1, M)$ is negative.

    \textbf{Case 2(b)(I) $0<m'<m_0$:} To show $\maxlg(m',M)$ (we drop $\maxlq$ from the notation for brevity) is negative for any $M>15$ in this sub-case in two steps. First we define a function $\maxlgprime$ and show that it is an upper bound on $\maxlg$. Then, we show that $\maxlgprime$ is negative in the range $0<m'<m_0$ by a careful analysis of its maximum value in this range, which completes the proof of \textbf{Case 2(b)(I)}. We now proceed with these steps.

    \textbf{Step (i):} Consider the function defined as $\maxlgprime(m', M)\triangleq\maxlq^{m'+1}+\frac{m'}{m'+1}\left(\frac{M}{M-1}\right)-1$. 
    We now prove that $\maxlgprime(m',M)$ is an upper-bound on function $\maxlg(m',M)$ for all $m'\in(0,m_0)$ and $M>15$. 
    \begin{align}
        \maxlg(\maxlq,m',M)-\maxlgprime(\maxlq,m',M)
        &= \frac{\maxlq^{m'+1}-\maxlq^{M}}{1-\maxlq^{M-1}}-\maxlq^{m'+1}\nonumber\\
        &= \frac{\maxlq^{m'+1}-\maxlq^{M}-\maxlq^{m'+1}(1-\maxlq^{M-1})}{1-\maxlq^{M-1}}\nonumber\\
        &= \frac{\maxlq^{M}(\maxlq^{m'}-1)}{1-\maxlq^{M-1}}.\label{eqn:MaxL_aprox_g_gprime}
    \end{align}

    It is easy to see that $\maxlq^{m'}-1<0$ for $m'\in(0,m_0)$. Hence, from \eqref{eqn:MaxL_aprox_g_gprime}, we have $\maxlg(\maxlq,m',M)<\maxlgprime(\maxlq,m',M)$. 
    
    \textbf{Step (ii):} Now, we show that $\maxlgprime$ is negative for $m'\in(0,m_0)$ for any $M>15$. We first observe that $\maxlgprime$ is decreasing in $M$, since $\frac{\partial \maxlgprime}{\partial M}=-\frac{m'}{(m'+1)(M-1)^2}<0$. Therefore, it is sufficient to prove that $\maxlgprime(m',16)$ is negative for $m'\in(0,m_0)$. 
    
    To show that $\maxlgprime(m',16)$ is negative for $m'\in(0,m_0)$, we first show that the derivative of $\maxlgprime$ with respect to $m'$ is non-negative in $m'\in(0,m_0)$, and thus it is sufficient to show $\maxlgprime(m_0,16)$ is negative.
    
    The derivative of $\maxlgprime(m',16)$ with respect to $m'$ is as follows:
    \begin{align}
        \frac{\partial \maxlgprime(m',16)}{\partial m'}
        &= \maxlq^{m'+1}\ln(\maxlq)+\frac{1}{(m'+1)^2}\frac{16}{15}\label{eqn:MaxL_partial_g'_w.r.t_m}
        % &= -\frac{\ln(1/\maxlq)}{(m'+1)^2}\left((m'+1)^2\maxlq^{m'+1}-\frac{16}{15\ln(1/\maxlq)}\right).
    \end{align}

    The second derivative of $\maxlgprime$ with respect to $m'$ is
    \begin{align}
        \frac{\partial^2 \maxlgprime(m',16)}{\partial m'^2} &= \maxlq^{m'+1}(\ln(\maxlq))^2-\frac{2}{(m'+1)^3}\frac{16}{15}\nonumber\\
        &=\frac{(\ln(\maxlq))^2}{(m'+1)^3}\left((m'+1)^3\maxlq^{m'+1}-\frac{32}{15(\ln(\maxlq))^2}\right).\label{eqn:MaxL_second_derivative_g'_w.r.t_m}
    \end{align}

    The sign of \eqref{eqn:MaxL_second_derivative_g'_w.r.t_m} is determined by the term in parentheses. Now, the term $\maxlk(m')\triangleq\maxlq^{m'+1}(m'+1)^3$ has a  derivative with respect to $m'$:
    \begin{equation*}
        \frac{\partial \maxlk(m')}{\partial m'}=\maxlq^{m'+1}(m'+1)^2(3+(m'+1)\ln(\maxlq)), 
    \end{equation*}
    which is decreasing in $m'$. 
    At $m'=m_0=\frac{2}{\ln(1/\maxlq)}-1$, we have the minimum value of $3-(m'+1)\ln(1/\maxlq)$, which is $3-(m_0+1)\ln(1/\maxlq)=3-\frac{2}{\ln(1/\maxlq)}\ln(1/\maxlq)=3-2=1>0$. Thus, $\frac{\partial \maxlk(m')}{\partial m'}>0$ for $m'\in(0,m_0)$, and hence $\maxlk(m')$ is increasing in $m'\in(0,m_0)$. It is straightforward to verify that $\maxlk(m_0)=\frac{32}{(4\exp(2)\ln(1/\maxlq))(\ln(1/\maxlq))^2}<\frac{32}{15(\ln(\maxlq))^2}$ (because $4\exp(2)\ln(1\maxlq)>15$). Therefore, from \eqref{eqn:MaxL_second_derivative_g'_w.r.t_m}, we see that $\frac{\partial^2 \maxlgprime(m',16)}{\partial m'^2}$ is negative from $m'\in(0,m_0)$. Thus, $\frac{\partial \maxlgprime(m',16)}{\partial m'}$ is decreasing in $m'\in(0,m_0)$. Thus, the minimum value of $\frac{\partial\maxlgprime}{\partial m'}$ in $m'\in(0,m_0)$ occurs at $m'=m_0$, which is $\frac{\partial \maxlgprime(m',16)}{\partial m'}\biggl|_{m'=m_0}=6.6092\times10^{-8}>0$.
    
    Thus, $\maxlgprime(m',16)$ is increasing in $m'\in(0,m_0)$. Its largest value is attained at $m_0$, which we calculate as follows:
    \begin{align}
        \maxlgprime(m_0,16) &= \maxlq^{m_0+1}+\frac{m_0}{m_0+1}\cdot\frac{16}{15}-1\nonumber\\
        &=0.601994^{2.9408+1}+\frac{2.9408}{2.9408+1}\cdot\frac{16}{15}-1\nonumber\\
        &=0.1353+0.7960-1\nonumber\\
        &=-0.0687<0.\label{eqn:MaxL_gprime_at_m0}
    \end{align}
    As, $\maxlgprime(m',M)<\maxlgprime(m_0,M)<\maxlgprime(m_0,16)$ for all $m'\in(0,m_0)$ and for all $M>15$, we have $\maxlgprime(m',16)$ is negative for $m'\in(0,m_0)$ and by \eqref{eqn:MaxL_aprox_g_gprime}, which completes the proof of \textbf{Case 2(b)(I)}.

    \textbf{Case 2(b)(II) $m'\in(m_0,M-1)$:} By Observation~\ref{obs:structureofgmaxlatm1m2_Cmaxl<fmaxl-sub-case}, we know that $\maxlg$ attains maximum value at $m'=m_1$. By \textbf{Case~2(b)(I)} we know that $\maxlg(m_1)<0$, as $m_1<m_0$. Further, Observation~\ref{obs:structureofgmaxlatm1m2_Cmaxl<fmaxl-sub-case} tells us that $\maxlg$ is decreasing in $m'$ in $m'\in(m_1,m_2)$. Note that $m_0\in(m_1,m_2)$, thus $\maxlg(m')<\maxlg(m_1)<0$ for any $m'\in[m_0,m_2]$.
    
    Now, in the interval $(m_2,M-1]$, $\maxlg$ is increasing as mentioned in Observation~\ref{obs:structureofgmaxlatm1m2_Cmaxl<fmaxl-sub-case}. However, by the definition of $\maxlg$ as in \eqref{eqn:MaxL_g_definition}, we have $\maxlg(M-1,M)=0$. Combining these two observations, we confirms that $\maxlg(m')<0$ for $m'\in(m_2,M-1)$ as well. Thus, $\maxlg(m')<0$ for all $m'\in(m_0,M-1)$, which completes the proof of \textbf{Case 2(b)(II)}.
    
    Combining all the above cases, we conclude that $\maxlg(q,m',M)<0$ for all $m'\in(0,M-1)$, for any $M>2$, and for all $q$ ranging $0\le q\le\maxlq=0.601994$. Hence, the proof of Theorem~\ref{thm:new_WBU_MaxL_thm} is complete. \hfill \qedsymbol
    
\begin{remark}
    It appears that Theorem~\ref{thm:new_WBU_MaxL_thm} holds for threshold $K/N\le0.602118$ but proving this analytically seems quite difficult.
\end{remark}

\begin{remark}
\label{rem:WSJ_optimality_comparison}
It is worth noting that, in the replicated, collusion-free setting, the $\wsj$ scheme corresponds to the $\wbu$ scheme with $K=1$. Consequently, the optimality of the $\wsj$ rate-privacy trade-off under the $\mathsf{MaxL}$ metric follows directly from Theorem~\ref{thm:new_WBU_MaxL_thm}, as the threshold code-rate condition, $K/N\le0.60199$ is strictly satisfied for all $N\ge2$ (since $1/N\le0.5<0.60199$). This result formally establishes the optimality of the trade-off in  \cite[Theorem 2]{WSJPIR2024} and confirms that the scheme matches the state-of-the-art trade-off reported in \cite{IWPIR2022}.
%It is worth noting that for the replicated, collusion-free setting, the rate-privacy trade-off of the $\wsj$ scheme under the $\mathsf{MaxL}$ metric as in \cite[Theorem 2]{WSJPIR2024} matches the trade-off achieved by the WPIR scheme reported in \cite{IWPIR2022}, and this is known to be the state of the art. \cite[Theorem 2]{WSJPIR2024} adds to the belief that perhaps this could be the information-theoretically optimal trade-off. 
\end{remark}
    
\section{Conclusion}
\label{sec:conclusion}
This work investigated the capacity of Sun-Jafar-type Weak PIR schemes. We established the optimality of the schemes proposed in \cite{WSJPIR2024} for the replicated, non-colluding scenario under both $\mathsf{MIL}$ and $\mathsf{MaxL}$ metrics. Furthermore, we determined specific thresholds for the storage rate ($K/N$) and collusion ratio ($T/N$) below which these schemes remain optimal for MDS-coded and $T$-collusion settings, respectively. Our results highlight that while the distributions supported on endpoints ($M'=0$ and $M'=M-1$) are optimal for low code rates, intermediate values of $M'$ become necessary to maximize the rate as the storage efficiency $K/N$ (or the ratio $T/N$ for $T$-collusion) increases beyond the derived thresholds. The exact characterization of the optimal rate-privacy achieved by the Sun-Jafar-type and the Banawan-Ulukus-type WPIR schemes in this regime is ongoing. However, the question of tight information-theoretic converses for the class of all WPIR schemes remains wide open, with only few results existing in literature. 

% \newpage
% \IEEEtriggeratref{17}
\bibliographystyle{IEEEtran}
% \bibliography{Converses_ArXiv_Jan152026.bib}
\bibliography{ref.bib}
%%%%
% \newpage
\appendices
\section{Useful Lemmas}
\label{appendix:usefullemmas}

\begin{lemma}
\label{lemma:f_increases_with_x_WBU_MIL_inter}
    For $a, b$ and $x$ be three positive reals such that $0<a<b$ and $x\in(0,1)$. Consider a function defined as $$h(x)=\frac{1-x^a}{1-x^b}.$$ Then the function $h(x)$ is increasing in $x\in(0,1)$.
\end{lemma}
\begin{IEEEproof}
    We will start by differentiating $h(x)$ with respect to $x$. Then, given constraints on $a$, $b$, and $x$, we will prove the desired result.
    \begin{align}
        g'(x)&=\frac{(1-x^b)(-ax^{a-1})-(1-x^a)(-bx^{b-1})}{(1-x^b)^2}\nonumber\\
        &=\frac{abx^{a-1}x^{b-1}}{(1-x^b)^2}\left(-\frac{1-x^a}{ax^{a-1}}+\frac{1-x^b}{bx^{b-1}}\right)\nonumber\\
        &=\frac{xabx^{a-1}x^{b-1}}{(1-x^b)^2}\left(-\frac{1}{a}\left(\frac{1}{x^{a}}-1\right)+\frac{1}{b}\left(\frac{1}{x^{b}}-1\right)\right)\nonumber\\
        &=\frac{abx^{a+b-1}}{(1-x^b)^2}\left(\frac{1}{a}\left(1-\frac{1}{x^{a}}\right)-\frac{1}{b}\left(1-\frac{1}{x^{b}}\right)\right)\label{eqn:lemma_g_inc}
    \end{align}
    
    Since $0<a<b$, then we can have i) $\frac{1}{a}>\frac{1}{b}$ and ii) $1-\frac{1}{x^a}>1-\frac{1}{x^b}$ (note that $x\in(0,1)$ therefore for $0<a<b$ we have $x^a>x^b$). Plugging these in \eqref{eqn:lemma_g_inc} and leading all terms are positive, thus we get $g'(x)>0$.
\end{IEEEproof}

\end{document}